\documentclass[twoside,12pt,reqno]{amsart}
\usepackage[foot]{amsaddr}
\usepackage[unicode=true,pdfusetitle, bookmarks=true,bookmarksnumbered=false,bookmarksopen=false, breaklinks=false,backref=false,colorlinks=false]{hyperref}
\hypersetup{colorlinks,linkcolor=myurlcolor,citecolor=myurlcolor,urlcolor=myurlcolor}
\usepackage{graphics,epstopdf,graphicx,amsthm,amsmath,amssymb,mathptmx,braket,colortbl,color,bm,framed,mathrsfs}
\usepackage[T1]{fontenc}
\usepackage[up]{subfigure}
\usepackage[font=small]{caption}
\usepackage{tikz}
\usepackage{pgfplots}
\pgfplotsset{width=10cm,compat=1.9}

\usepackage{float}
\usepackage{amsfonts}
\usepackage{cite}
\usepackage{hyperref}
\hypersetup{
     colorlinks = true,
     linkcolor = black,
     anchorcolor = brown,
     citecolor = purple,
     filecolor = brown,
     urlcolor = purple
}

\usepackage[left=4cm, right=4cm, top=2cm]{geometry}

\usepackage[normalem]{ulem}

\usepackage{xcolor,cancel}

\def\mod{\mathrm{mod}\ }

\def\affil#1{\texorpdfstring{$^{#1}$}{}}

\def\Cost{\mathrm{Cost}}
\def\CiS{\mathrm{CiS}}
\def\supp{\mathrm{supp}}
\def\SWAP{\mathrm{SWAP}}
\def\SU{\mathrm{SU}}

\usepackage{multirow,tabu}
\usepackage{array}
\usepackage{tabularx}
\newcolumntype{L}[1]{>{\raggedright\let\newline\\\arraybackslash\hspace{0pt}}m{#1}}
\newcolumntype{C}[1]{>{\centering\let\newline\\\arraybackslash\hspace{0pt}}m{#1}}
\newcolumntype{R}[1]{>{\raggedleft\let\newline\\\arraybackslash\hspace{0pt}}m{#1}}
\usepackage{tabu}

\usepackage{qcircuit}

\definecolor{myurlcolor}{rgb}{0,0,0.9}
\newcommand{\be}{\begin{equation}}
\newcommand{\ee}{\end{equation}}
\newcommand{\beq}{\begin{eqnarray}}
\newcommand{\eeq}{\end{eqnarray}}
\newcommand{\beqs}{\begin{eqnarray*}}
\newcommand{\eeqs}{\end{eqnarray*}}

\newcommand{\proj}[1]{| #1\rangle\!\langle #1 |}

\newcommand{\inner}[2]{\langle #1 , #2\rangle}

\DeclareMathOperator{\trace}{Tr}
\newcommand{\Ptr}[2]{\trace_{#1}\Pa{#2}}
\newcommand{\Tr}[1]{\Ptr{}{#1}}

\newcommand{\Pa}[1]{\left[#1\right]}

\newcommand{\norm}[1]{\left\lVert #1 \right\rVert}

\theoremstyle{plain}
\newtheorem{thm}{Theorem}
\newtheorem{lem}[thm]{Lemma}
\newtheorem{prop}[thm]{Proposition}
\newtheorem{cor}[thm]{Corollary}
\newtheorem{con}[thm]{Conjecture}
\newtheorem{defn}[thm]{Definition}
\theoremstyle{definition}
\newtheorem{rema}{Remark}
\newtheorem{example}{Example}

\tikzstyle WL=[line width=10pt,opacity=1.0]
\usetikzlibrary{arrows.meta}

\newcommand*{\myproofname}{Proof}

\def\ot{\otimes}
\def\complex{\mathbb{C}}
\def\real{\mathbb{R}}

\def\poly{\mathrm{poly}}
\def\i{\mathrm{i}}
\def\d{\mathrm{d}}
\def\e{\mathrm{e}}

\makeatother

\begin{document}

\title{Complexity of quantum circuits via sensitivity, magic, and coherence}

\author{Kaifeng Bu\affil{1}}

\author{Roy J. Garcia\affil{1}}

\author{Arthur Jaffe\affil{1}}

\author{Dax Enshan Koh\affil{2}}

\author{Lu Li\affil{3}}

\address[1]{\textnormal{Harvard University, Cambridge, Massachusetts 02138, USA.
Email: 
\texttt{kfbu@fas.harvard.edu}; \texttt{roygarcia@g.harvard.edu};
\texttt{arthur\textunderscore jaffe@harvard.edu}
}}
\address[2]{\textnormal{Institute of High Performance Computing, Agency for Science, Technology and Research (A*STAR), 1 Fusionopolis Way, \#16-16 Connexis, Singapore 138632, Singapore.
Email: \texttt{dax\textunderscore koh@ihpc.a-star.edu.sg}
}}
\address[3]{\textnormal{Department of Mathematics, Zhejiang Sci-Tech University, Hangzhou, Zhejiang 310018, China.
Email: \texttt{lilu93@zju.edu.cn}
}}

\begin{abstract}
Quantum circuit complexity---a measure of the minimum number of gates needed to implement a given unitary transformation---is a fundamental concept in quantum computation, with widespread applications ranging from determining the running time of quantum algorithms to understanding the physics of black holes. In this work, we study the complexity of quantum circuits using the notions of sensitivity, average sensitivity (also called influence), magic, and coherence. We characterize the set of unitaries with vanishing sensitivity and show that it coincides with the family of matchgates. Since matchgates are tractable quantum circuits, we have proved that sensitivity is necessary for a quantum speedup. As magic is another measure to quantify quantum advantage, it is interesting to understand the relation between magic and sensitivity. We do this by introducing a quantum version of the Fourier entropy-influence relation. Our results are pivotal for understanding the role of sensitivity, magic, and coherence in quantum computation.

\end{abstract}

\maketitle

\setcounter{tocdepth}{1}
\tableofcontents
\section{Introduction}
A central problem in the field of quantum information and computation is to compute the complexity required to implement a target unitary operation $U$. One usually defines this to be the minimum number of basic gates needed to synthesize $U$ from some initial fiducial state\cite{nielsen2010quantum,kitaev2002classical,aaronson2016complexity}.
To determine the so-called \textit{quantum circuit complexity} of a given unitary operation, a closely related concept, called the \textit{circuit cost}, was proposed and investigated in a series of seminal papers by Nielsen et al.~\cite{nielsen2006geometric,nielsen2006optimal,nielsen2006quantum, dowling2008geometry}. Surprisingly, the circuit cost, defined 
as the minimal geodesic distance between the target unitary operation and the identity operation in some curved geometry, was shown to provide a useful lower bound for the quantum circuit complexity~\cite{nielsen2006quantum,nielsen2006optimal}. 

In more recent years, the quantum circuit complexity, as well as the circuit cost, was shown to also play an important role in the domain of high-energy physics~\cite{brown2017quantum,susskind2016typical,brown2016holographic,chapman2018toward,brandao2021models}. For example, its evolution was found to exhibit identical patterns to how the geometry hidden inside black hole horizons evolves. Further studies have also investigated the circuit complexity in the context of quantum field theories \cite{jefferson2017circuit,takayanagi2018holographic,bhattacharyyaj2018circuit}, including conformal field theory \cite{chagnet2022complexity, bhattacharyya2022complexity} and topological quantum field theory \cite{couch2021circuit}. Recently, Brown and Susskind argue that the property of possessing less-than-maximal entropy, or \textit{uncomplexity}, could be thought of as a resource for quantum computation \cite{brown2017quantum}. This was supported by Yunger Halpern~et~al.\ who present a resource theory of quantum uncomplexity~\cite{halpern2021resource}. Furthermore, a connection between quantum entanglement and quantum circuit complexity was revealed by Eisert, who proved that the entangling power of a unitary transformation provides a lower bound for its circuit cost~\cite{eisert2021entangling}.  

Let us summarize the main ideas we present in this paper, which we will describe in more detail in \S\ref{sec:Results}.
In this paper, we study the quantum circuit complexity of quantum circuits via their sensitivities, magic, and coherence. The first property, namely sensitivity, is a measure of complexity that plays an important role in the analysis of Boolean functions \cite{odonnell2014analysis,kahn1988influence} and can be applied to a range of topics, including the 
circuit complexity of Boolean circuits \cite{linial1989constant,boppana1997average,jukna2012boolean}, error-correcting codes~\cite{lovett2011bounded}, and quantum query complexity~\cite{shi2000lower}. A fundamental result in circuit complexity is that the average sensitivity, also called the \textit{influence}, of constant-depth Boolean circuits is bounded above by the depth and the number of gates in the 
circuit~\cite{linial1989constant,boppana1997average}. While the notion of influence has been generalized to describe quantum Boolean functions
\cite{montanaro2010quantum}, considerably little is hitherto known  about the connection between the sensitivity (or influence) and the circuit complexity of a 
quantum circuit. In this regard, our first result provides an upper bound on the circuit sensitivity---a measure of sensitivity for unitary transformations---of a quantum circuit by its circuit cost.

Secondly, we characterize
unitaries with zero circuit sensitivity, which we call \textit{stable unitaries}.
We generalize the definition of sensitivity to Clifford algebras, where we use the noise operator defined by 
Carlen and Lieb~\cite{carlen1993optimal}. We find that
stable gates in this case are exactly matchgates, a well-known family of tractable quantum circuits~\cite{valiant2002quantum,bravyi2005lagrangian,divincenzo2004fermionic,terhal2002classical,jozsa2008matchgates,brod2016efficient,hebenstreit2019all}.
This provides a new understanding of matchgates via sensitivity. 
Our result also implies that sensitivity is necessary for a quantum computational advantage; for a more extended discussion, see Remark \ref{rema:tradeoff}.
In addition, we show a relation between average scrambling and the average sensitivity.

Magic is another important resource 
in quantum computation, which characterizes how far away a quantum state (or gate) is from the set of stabilizer states (or gates). The Gottesman-Knill theorem~\cite{gottesman1998heisenberg} states
that stabilizer circuits comprising Clifford unitaries and stabilizer inputs and measurements can be simulated efficiently on a classical computer. 
Hence, magic is necessary to realize a quantum advantage~\cite{nest2010classical,jozsa2014classical, koh2017further, bouland2018complexity,yoganathan2019quantum}. Magic measures have been used to bound the classical simulation time in quantum computation~\cite{bravyi2016trading,bravyi2019simulation,howard2017application,seddon2021quantifying,seddon2019quantifying,wang2019quantifying,bu2019efficient,bu2022classical}, and also in condensed matter physics~\cite{liu2020many}. However, the relationship between magic and the complexity of quantum circuits has so far largely been unexplored.

To reveal the connection between magic and circuit complexity, we implement two different approaches.  The first approach (see \S\ref{sec:mag_1}) uses consequences of the quantum Fourier entropy-influence relation and conjecture, 
which shows the relation between magic and  sensitivity.
It can be summarized by the set of inferences diagrammed here:

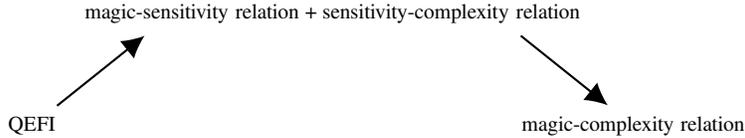
\tikzstyle{block} = [rectangle]
\tikzstyle{arrow} = [>={LaTeX[width=3mm,length=3mm]},thick, ->]
\begin{center}
\begin{tikzpicture}[auto]
   \node [block] (top) at (0,0) {\tiny magic-sensitivity relation + sensitivity-complexity relation};
   \node [block] (left) at (-4,-1.5) {\tiny QEFI};
   \node [block] (right) at (4,-1.5) {\tiny magic-complexity relation};
   \draw [arrow] (left) -- (-2.5,-0.3);
   \draw [arrow] (2.5,-0.3) -- (right);
\end{tikzpicture}
\end{center}
  
Depending on whether one takes a proven result or a conjectured bound, one arrives at an uninteresting or interesting result, respectively. The classical Fourier entropy-influence conjecture was proposed by Friedgut and Kalai~\cite{friedgut1996every}, and has many useful implications in the analysis of Boolean functions and computational learning theory. For example, if the  Fourier entropy-influence conjecture holds, then it implies the existence of a polynomial-time agnostic
learning algorithm for disjunctive normal forms (DNFs)~\cite{mansour1994learning}.

The second method (see \S\ref{sec:mag_2}) we take here is to exhibit the connection between magic and circuit cost directly by introducing the magic rate and magic power. Magic power quantifies the 
incremental magic by the circuit, while the magic rate quantifies the small incremental magic in infinitesimal time.

Finally, we show the connection between coherence and circuit complexity for quantum circuits. Quantum coherence, which 
arises from superposition, plays a fundamental role in 
quantum mechanics. 
The recent significant developments in 
quantum thermodynamics~\cite{lostaglio2015quantum,lostaglio2015description} and quantum biology \cite{plenio2008dephasing,lloyd2011quantum,levi2014quantitative}
have shown that coherence can be a very useful resource at the nanoscale. This has led 
to the development of the resource theory of coherence~\cite{aberg2006quantifying, baumgratz2014quantifying,winter2016operational,bu2017maximum,Streltsov2017colloquium,bischof2019resource}. However, thus far, little is known about the connection between coherence and circuit complexity. In this paper, we address this gap and provide a lower bound on the circuit cost by the power of coherence in the circuit. 

The rest of the paper is structured as follows. 
In \S\ref{sec:Results}, we summarize the main results of our work. 
In \S\ref{sec:influence}, we 
investigate the connection between circuit complexity and circuit sensitivity and propose a new interpretation of matchgates in terms of sensitivity. 
In  \S\ref{sec:FEn_inf}, we consider the relationship between 
quantum Fourier entropy and influence.
In \S\ref{sec:magic}, we study the connection between magic and the circuit cost of quantum circuits.
In \S\ref{sec:coh}, we study the connection between 
coherence and the circuit cost of quantum circuits.

\subsection{Main results}
\label{sec:Results}
We start by summarizing three of our main results concerning lower bounds on quantum circuit complexity in terms of average sensitivity, magic, and 
coherence. Here, the complexity of a quantum circuit is taken to be the circuit cost introduced by Nielsen et al.: 
\begin{defn} [Nielsen et al.~\cite{nielsen2006quantum}]
Let $U\in \SU(d^n)$ be a unitary operation and $h_1, \ldots, h_m$ be traceless Hermitian operators that are supported on 2 qudits and normalized as $\norm{h_i}_{\infty}=1$. The circuit cost of $U$, with respect to $h_1, \ldots, h_m$, is defined as
\begin{eqnarray}
\Cost(U):=\inf\int^1_0
\sum^m_{j=1}|r_j(s)|\d s.
\end{eqnarray}
where the infimum above is taken over all continuous functions $r_j:[0,1]\to \mathbb R$ satisfying
\begin{eqnarray}
U=\mathcal P\exp\left(-\i\int^1_0H(s)\d s\right),
\label{eq:path-orderedU}
\end{eqnarray}
and 
\begin{eqnarray}\label{eq:rj}
H(s)=
\sum^m_{j=1} r_j(s)h_j,
\end{eqnarray}
where $\mathcal P$ denotes the path-ordering operator.

\end{defn}

The theorem below, which gives lower bounds for the circuit cost, collects Theorems \ref{thm:cost_In}, \ref{thm:cost_ma} and \ref{thm:cost_co} in one place:

\begin{thm}[\bf Results on Circuit Complexity] 
The circuit cost of a quantum circuit $U\in \SU(d^n)$ is
lower bounded as follows: 
\begin{eqnarray}
\Cost(U)
\geq c
\max\left\{\CiS[U], \frac{\mathcal{M}[U]}{d^2}, \frac{\mathcal{C}_r(U)}{\log (d)}\right\},
\end{eqnarray}
where $c$ is a universal constant independent of $d$ and $n$. The quantities $\CiS[U]$ ($\mathcal{M}[U]$, $\mathcal{C}_r(U)$, respectively), 
defined formally in \eqref{eq:circuit_sensitivity} (\eqref{eq:magic_power}, \eqref{eq:cohering_power}, respectively),
quantify the 
sensitivity (magic, coherence, respectively) of quantum circuits. Note that here and throughout this paper, the logarithm is taken to be of base 2. 
\end{thm}

We also define the circuit sensitivity $\CiS^G$ for any unitary in terms of the generators of the Clifford algebra, yielding a new understanding for matchgates (see Theorem \ref{thm:cha_match} for more details): 
  \begin{thm} [\bf  Matchgates via Sensitivity]
  A unitary $U$ satisfies $\CiS^G[U]=0$ if and only if it is a matchgate. 
\end{thm}
Matchgates are a well-known family of tractable circuits,  and our result shows that $\CiS^G$ could also be used to serve as a measure of non-Gaussianity (noting that matchgates are also called Gaussian operations).

To show the connection between magic and influence (or non-Gaussianity quantified by influence), we also prove the following statement (an informal version of Theorem \ref{thm:main1}): 

\begin{thm}[\bf Quantum Fourier Entropy-Influence Relation]
For any linear $n$-qudit operator $O$  with $\norm{O}_2=1$, we have 
\begin{eqnarray*}
H[O]\leq c[\log n+\log d]I[O]+h[P_O[\vec{0}]],
\end{eqnarray*}
where $h(x):=-x\log x-(1-x)\log(1-x)$ is the binary entropy and $c$ is a universal constant. 
\end{thm}

\section{Sensitivity and circuit complexity}\label{sec:influence}

Given the $n$-qudit system $\mathcal{H}=(\complex^d)^{\ot n}$, the inner product between two operators $A$ and $B$ on $\mathcal H$ is defined as
$
\inner{A}{B}=
\frac{1}{d^n}
\Tr{A^\dag B}
$, and the $l_2$ norm induced by the inner product is defined by $\norm{A}_2:=\sqrt{\inner{A}{A}}$. 
More generally, for $p\geq 1$, the $l_p$ norm is defined as $\norm{A}_{p}=(\frac{1}{d^n}\Tr{|A|^p})^{1/p}$ with $|A|=\sqrt{A^\dag A}$.
Taking $V:=\mathbb{Z}_d\times \mathbb{Z}_d$, 
the set of generalized Pauli operators is 
\begin{eqnarray}
\mathcal{P}_n=\set{P_{\vec{a}}:P_{\vec{a}}=\ot_i P_{a_i}}_{\vec{a}\in V^n},
\end{eqnarray}
where  $P_{a_i}=X^{s_i}Z^{t_i}$ for any $a_i=(s_i,t_i)\in V$. Here, the qudit Pauli $X$ and $Z$ are the shift and clock operators, respectively, defined by $X\ket{j}=\ket{j+1 \ (\mod d)}$ and $Z\ket{j}=\exp(2\i j\pi/d)\ket{j}$, respectively. 
Let us define $P_O[\vec{a}]$ for any $\vec{a}\in V^n$ as
\begin{eqnarray}\label{eq:proA}
P_O[\vec{a}]=\frac{1}{d^{2n}}
|\Tr{OP_{\vec{a}}}|^2,~~~\forall \vec{a}\in V^n.
\end{eqnarray}
Note that the condition  $\norm{O}_2=1$ is equivalent to saying that $\set{P_O[\vec{a}]}_{\vec{a}}$ is
a probability distribution over $V^n$.

\subsection{Influence}

\begin{defn}[Montanaro and Osborne~\cite{montanaro2010quantum}]
Given a linear operator $O$, the local influence at the
$j$-th qudit is defined as 
\begin{eqnarray}
I_j[O]
=\sum_{\vec{a}:a_j\neq (0,0)}
P_O[\vec{a}],
\end{eqnarray}
and the total influence is defined as the 
sum of all the local influences:
\begin{eqnarray}
I[O]=
\sum_{j\in [n]}
I_j[O].
\end{eqnarray}
\end{defn}
With the assumption that $P_O$ in \eqref{eq:proA} is a probability distribution, the local influence
and total influence can be rewritten, respectively, as
\begin{eqnarray}
I_j[O]
&=&\sum_{\vec{a}:a_j\neq (0,0)}
P_O[\vec{a}]=
\mathop{\mathbb{E}}_{\vec{a}\in P_O}|a_j|,\\
\label{eq:In}I[O]&=&\sum_{\vec{a}\in V^n}
|\supp(\vec{a})|P_O[\vec{a}]
=\mathop{\mathbb{E}}_{\vec{a}\sim P_O}
|\vec{a}|,
\end{eqnarray}
where $|a_j| = 1$ if $a_j=(0,0)$ and 0 otherwise; $\supp(\vec{a})$ (the support of $\vec{a}$) denotes the set of indices $i$ for which $a_i\neq 0$; and  $|\vec{a}|:=|\supp(\vec{a})|$.

Note that it is easy to see that the 
influence can be used to quantify the sensitivity of the single-qudit depolarizing channel 
$D_\gamma(\cdot)=(1-\gamma)(\cdot)+\gamma\Tr{\cdot}\mathbb{I}/d$ as follows
\begin{eqnarray}
\frac{\partial}{\partial\gamma}
\norm{D^{(j)}_\gamma[O]}^2_2\Big|_{\gamma=0}
=-2I_j[O],
\end{eqnarray}
where $D^{(j)}_\gamma$ denotes the depolarizing channel acting on 
the $j$-th qudit. This implies that 
\begin{eqnarray}
\frac{\partial}{ \partial \gamma}
\norm{D^{\ot n}_\gamma[O]}^2_2\Big|_{\gamma=0}
=-2I[O].
\end{eqnarray}
Hence, influence is an average version of sensitivity  with respect to depolarizing noise. Note that the notion of influence, $I_j(O)$ and $I(O)$, 
could be applied to quantum states $\ket{\psi}$ by setting $O=\sqrt{d^n}\proj{\psi}$ to ensure that the corresponding probability distribution $P_O$ defined in \eqref{eq:proA} sums to 1.

\subsection{Circuit sensitivity and complexity}

\begin{defn}[{\textbf{Circuit Sensitivity}}] For a unitary $U$, 
the circuit sensitivity $\CiS[U]$ is  the change of influence caused by $U$, defined as
\begin{eqnarray}
\CiS[U]=
\max_{O:\norm{O}_2=1}
\left|I[U OU^\dag]-I[O]\right|.
\label{eq:circuit_sensitivity}
\end{eqnarray}
\label{def:circuit_sensitivity}
\end{defn}

First, let us present a basic lemma of circuit sensitivity, which indicates that in the maximization in \eqref{eq:circuit_sensitivity}, it suffices to just consider traceless operators:

\begin{lem}
The circuit sensitivity equals 
\begin{eqnarray}
\CiS[U]=
\max_{O:\norm{O}_2=1,\Tr{O}=0}
\left| I[U OU^\dag]-I[O]\right|,
\end{eqnarray}
that is, it suffices to just consider a maximization over all traceless operators with $\norm{O}_2=1$.

\end{lem}
\begin{proof}
First, $P_O[\vec{0}]$ defined in \eqref{eq:proA} is unitarily invariant. Hence, if $\Tr{O}\neq 0$, let us define a new operator $O'$ as
\begin{eqnarray*}
O'=\frac{1}{\sqrt{1-P_O[\vec{0}]}}
\left(O-\frac{\Tr{O}}{d^n}I\right).
\end{eqnarray*}
Then $O'$ satisfies the conditions $\Tr{O'}=0$ and $\norm{O'}_2=1$.
Also,
\begin{eqnarray*}
I[O']&=&\frac{1}{1-P_O[\vec{0}]}
I[O],\\
I[UO'U^\dag]&=&\frac{1}{1-P_O[\vec{0}]}
I[UOU^\dag].
\end{eqnarray*}
Hence, we have 
\begin{eqnarray*}
I[UO'U^\dag]-I[O']
=\frac{1}{1-P_O[\vec{0}]}(I[UOU^\dag]-I[O]).
\end{eqnarray*}
Therefore, the maximum must be attained by
traceless operators.
\end{proof}

Now, let us consider the $n$-qudit Hamiltonian acting nontrivially on a
$k$-qudit subsystem. We 
prove here a simple upper bound on the 
total change of the total influence $I$ through unitary evolution. 
    \begin{prop}[\textbf{Small Total Circuit Sensitivity}]
Given an $n$-qudit system with 
a Hamiltonian $H$ acting nontrivially on a $k$-qudit subsystem, the total change of 
influence induced by the unitary $U_t=\e^{-\i tH}$  is bounded from above by $k$:
\begin{eqnarray}
\CiS[U_t]
\leq k.
\end{eqnarray}
\end{prop}
\begin{proof}

Since $H$ acts on only a $k$-qudit subsystem,  there exists a 
subset $S$ of size $k$ such that $H=H_S\ot I_{S^c}$ and $U_t=U_S\ot I_{S^c}$. 
Due to the subadditivity of the circuit sensitivity under tensorization (Proposition \ref{prop:zero-con}), 
$
\CiS[U_t]
\leq \CiS[U_S]
\leq k.
$
\end{proof}

Now, let us introduce the  \textit{influence rate} to quantify the change of influence 
in an infinitesimally small time interval. This will be used to prove the connection between circuit sensitivity and
circuit complexity. 

\begin{defn}[\textbf{Influence Rate}]
Given an $n$-qudit 
 Hamiltonian $H$ and a linear operator $O$ with $\norm{O}_2=1$,  the influence rate of the
 unitary $U_t=\e^{-\i t H}$ acting on $O$ is defined as follows
\begin{eqnarray}
R_I(H, O)
=\frac{\d I[U_t O U^\dag_t]}{\d t}\Big|_{t=0},
\end{eqnarray}
which can be used to quantify small incremental influence for 
a given unitary evolution. 
\end{defn}

By a direct calculation, we have the following 
explicit form of the influence rate:
\begin{eqnarray}\label{eq:inf-ra}
R_I(H, O)
=\frac{\i}{d^{2n}}
\sum_{\vec{a}\in V^n}|\vec{a}|
\left(\Tr{[O,H]P_{\vec{a}}}
\Tr{OP^\dag_{\vec{a}}}
+\Tr{[O,H]P^\dag_{\vec{a}}}
\Tr{OP_{\vec{a}}}
\right).
\end{eqnarray}

First, let us provide an upper bound on the influence rate.
\begin{lem}\label{lem:globa_inf_ra}
Given an $n$-qudit system with 
a Hamiltonian $H$ and a linear operator $O$ with $\norm{O}_2=1$, we have 
\begin{eqnarray}
|R_I(H,O)|
\leq 4n\norm{H}_{\infty},
\end{eqnarray}
where $\norm{H}_{\infty}$ denotes the operator norm.
\end{lem}
\begin{proof}
Since $|\vec a|\leqslant n$, the Schwarz inequality yields
\begin{eqnarray*}
\frac{1}{d^{2n}}
\sum_{\vec{a}\in V^n}|\vec{a}|
\left|\Tr{[O,H]P_{\vec{a}}}\right|
\left|\Tr{OP^\dag_{\vec{a}}}\right|
&\leq& n\frac{1}{d^{2n}}
\sum_{\vec{a}\in V^n}
|\Tr{[O,H]P_{\vec{a}}}|
\left|\Tr{OP^\dag_{\vec{a}}}\right|\\
&=&n
\norm{[O,H]}_2
\norm{O}_2
\leq 2n\norm{H}_{\infty},
\end{eqnarray*}
where the last inequality
comes from the 
H\"older inequality and the fact that $\norm{O}_2=1$.
Similarly, we can prove that 
\begin{eqnarray*}
\frac{1}{d^{2n}}
\sum_{\vec{a}\in V^n}|\vec{a}|
\left|\Tr{[O,H]P^\dag_{\vec{a}}}\right|
|\Tr{OP_{\vec{a}}}|
\leq 2n\norm{H}_{\infty}.
\end{eqnarray*}
Hence, by the expression of influence rate in 
\eqref{eq:inf-ra}, we have 
\begin{eqnarray*}
|R_I(H,O)|
\leq 4n\norm{H}_{\infty}.
\end{eqnarray*}
\end{proof}

Let us provide an upper bound on the influence rate for the unitary generated by a local Hamiltonian.

\begin{thm}[\textbf{Small Incremental Influence}]\label{thm:main_inf} Given an $n$-qudit system with the 
 Hamiltonian $H$ acting on a $k$-qudit subsystem, and 
a linear operator $O$ with unit norm $\norm{O}_2=1$, one has 
\begin{eqnarray}
|R_I(H, O)|\leq4 k\norm{H}_{\infty}.
\end{eqnarray}
\end{thm}
\begin{proof}
Since $H$  acts on a $k$-qudit subsystem, there exists a 
subset $S$ of size $k$ such that $H=H_S\ot I_{S^c}$. 
Define 
$O^{(1)}_{\vec{b}}$ on $(\complex^d)^{S^c}$ for $\vec{b}\in V^{S}$ by 
\begin{eqnarray}\label{eq:O1}
O^{(1)}_{\vec{b}}
=\frac{1}{d^{n-k}}
\text{Tr}_{S}[O P_{\vec{b}}]\;.
\end{eqnarray}
Also define
$O^{(2)}_{\vec{c}}$ on $(\complex^d)^{S}$ for any $\vec{c}\in V^{S^c}$ as 
\begin{eqnarray}\label{eq:O2}
O^{(2)}_{\vec{c}}
=\frac{1}{d^{n-k}}
\text{Tr}_{S^c}[O P_{\vec{c}}].
\end{eqnarray}
Note that 
$
\sum_{\vec{b}\in V^{S}}\norm{O^{(1)}_{\vec{b}}}^2_2=
\sum_{\vec{c}\in V^{S^c}}\norm{O^{(2)}_{\vec{c}}}^2_2=1
$.
Defining
$A_{\vec{b}}=O^{(1)}_{\vec{c}}/\norm{O^{(1)}_{\vec{c}}}_2$ and
$B_{\vec{c}}=O^{(2)}_{\vec{c}}/\norm{O^{(2)}_{\vec{c}}}_2$, we get that
$I[O]$ can be written as
\begin{eqnarray*}
I[O]=\sum_{\vec{c}\in V^{S^c}}
\norm{O^{(2)}_{\vec{c}}}^2_2
I[B_{\vec{c}}]
+\sum_{\vec{b}\in V^S}
\norm{O^{(1)}_{\vec{b}}}^2_2
I[A_{\vec{b}}].
\end{eqnarray*}
Hence,
 \begin{eqnarray*}
I [U_tO U^\dag_t]
 =\sum_{\vec{c}\in V^{S^c}}
\norm{O^{(2)}_{\vec{c}}}^2_2
I[U_tB_{\vec{c}}U^\dag_t]
+\sum_{\vec{b}\in V^S}
\norm{O^{(1)}_{\vec{b}}}^2_2
I[A_{\vec{b}}]\;,
 \end{eqnarray*}
and so  
\begin{eqnarray*}
R_I(H,O)
= \sum_{\vec{c}\in V^{S^c}}
\norm{O^{(2)}_{\vec{c}}}^2_2
R_I(H_S,B_{\vec{c}}).
\end{eqnarray*}
Since both $H_S$ and $B_{\vec{c}}$ for any $\vec{c}\in V^{S^c}$ act on a $k$-qudit subsystem, we have 
\begin{eqnarray*}
R_I(H_S,B_{\vec{c}})\leq 4k\norm{H_S}_{\infty}
\end{eqnarray*}
by 
 Lemma \ref{lem:globa_inf_ra}. Therefore, we obtain 
 \begin{eqnarray*}
 |R_I(H,O)|\leq 
 \sum_{\vec{c}\in V^{S^c}}
\norm{O^{(2)}_{\vec{c}}}^2_2
|R_I(H_S,B_{\vec{c}})|
\leq 4k\norm{H_S}_{\infty},
 \end{eqnarray*}
as claimed.
\end{proof}

Here, we use circuit sensitivity to quantify the average sensitivity of a quantum circuit. In classical Boolean circuits, 
the average sensitivity of the circuit plays an important role in lower bounding the complexity of a circuit \cite{linial1989constant,boppana1997average,jukna2012boolean}. Hence,
a natural question is: what is the connection between circuit sensitivity and circuit 
complexity for quantum circuits? Here, we use the circuit cost defined in \cite{nielsen2006quantum}
to quantify the  complexity of quantum circuits. Our next result establishes a connection between the circuit sensitivity and the circuit cost of a quantum circuit.

\begin{thm}[\textbf{Circuit Sensitivity Lower Bounds Circuit Cost}]\label{thm:cost_In}
The circuit cost of a quantum circuit $U\in \SU(d^n)$ is
lower bounded by the circuit sensitivity as follows
\begin{eqnarray}
\Cost(U)
\geq 
\frac{1}{8}
\CiS[U].
\end{eqnarray}
\end{thm}
\begin{proof}
The proof follows the same idea as that in 
\cite{marien2016entanglement, eisert2021entangling}. 
First, let us take 
a Trotter decomposition of $U$ such that for arbitrarily small $\epsilon>0$,
\begin{eqnarray*}
\norm{U-V_N}_{\infty}
\leq \epsilon,
\end{eqnarray*}
where $V_N$ is defined as 
follows
\begin{eqnarray*}
V_N&:=&\prod^N_{t=1}W_t,\\
W_t&:=&\exp\left(-\frac{\i}{N}\sum^m_{j=1}r_j\left(\frac{t}{N}\right)h_j\right).
\end{eqnarray*}
and 
\begin{eqnarray*}
W_t&=&\lim_{l\to\infty}W^{(l)}_t,\\
W^{(l)}_t&:=&\left(W^{1/l}_{t,1} \cdots W^{1/l}_{t,l}\right)^l,\\
W_{t,j}&:=&\exp\left(
-\frac{\i}{N}
r_{j}\left(
\frac{t}{N}
\right)h_j
\right).
\end{eqnarray*}
Let us define 
$O_t=W_tO_{t-1}W^\dag_t$ with $O_0=O$. Then by applying 
$W_t$, we have 
\begin{eqnarray*}
I[O_t]-I(O_{t-1})
&=&
I\left(W_tO_{t-1} W^\dag_t\right)-I(O_{t-1})\\
&=&\lim_{l\to\infty}
I\left(W^{(l)}_tO_{t-1}W^{(l)}_t \right)-I(O_{t-1})\\
&\leq& \frac{l}{N}
\sum^m_{j=1}\frac{8}{l}\left|r_j\left(\frac{t}{N}\right)\right|\\
&=& \frac{8}{N}
\sum^m_{j=1}\left|r_j\left(\frac{t}{N}\right)\right|,
\end{eqnarray*}
where the inequality above follows from Theorem \ref{thm:main_inf} for $k=2$. 
Taking the summation over all $t$, we have 
\begin{eqnarray*}
I(UO U^\dag)-I(O)
\leq \frac{8}{N}
\sum^N_{t=1}
\sum^m_{j=1}
\left|r_j\left(\frac{t}{N}\right)\right|.
\end{eqnarray*}
Since the circuit cost can be expressed as
 \begin{eqnarray*}
 \Cost(U)=
 \lim_{N\to \infty}\sum^N_{t=1}
\sum^m_{j=1}
\left|r_j\left(\frac{t}{N}\right)\right|,
 \end{eqnarray*}
we have 
 \begin{eqnarray*}
I(UO U^\dag)-I(O)
\leq 8
\Cost(U),
 \end{eqnarray*}
which completes the proof of the theorem.
\end{proof}

\subsection{Stable unitaries}
Here we characterize quantum circuits with zero circuit sensitivity and provide a complete characterization of such unitaries.

\begin{defn}
An $n$-qudit unitary (or gate or circuit) $U$ is stable if $\CiS[U]=0$.
\end{defn}

Here, to characterize the stable unitaries, we need to consider
weight-1 Pauli operators, i.e.\ $P_{\vec{a}}$ with $|\vec{a}|=1$.

\begin{prop}\label{prop:zero-con}
The circuit sensitivity satisfies the following three properties:

\begin{enumerate}
    \item 
An $n$-qudit unitary $U$ is stable if and only if for any weight-1 Pauli operator $O$, both $UOU^\dag$ and $U^\dag O U$ can be written as a
linear combination of weight-1 Pauli operators.
\item $\CiS[V_2UV_1]=\CiS[U]$ for any unitary $V_1$ and any stable unitary $V_2$.
\item 
$\CiS$ is subadditive under multiplication and tensorization:
\begin{eqnarray}
\CiS[UV]\leq \CiS[U]+\CiS[V],\quad
\CiS[U\ot V]\leq \CiS[U]+\CiS[V].
\end{eqnarray}
\end{enumerate}
\end{prop}
\begin{proof} \hfill

\noindent(1)
If $\CiS[U]=0$,  for any weight-1 Pauli operator $O$,
\begin{eqnarray*}
I[UOU^\dag]=I[O]=1.
\end{eqnarray*}
Hence, $UOU^\dag$  can be written   as a
linear combination of weight-1 Pauli operators. Similarly, $U^\dag OU$  can be written as a
linear combination of weight-1 Pauli operators.

On the other hand, if it holds that for any weight-1 Pauli operator $O$, both $UOU^\dag$ and $U^\dag OU$ can be written as a
linear combination of weight-1 Pauli operators, then 
$UP_{\vec{a}} U^\dag$  and $UP_{\vec{a}} U^\dag$  can be written as a linear combination 
of Pauli operators with weights less than $|\vec{a}|$. Hence,  we have 
\begin{eqnarray}\label{con:eq_weight}
\Tr{P^\dag_{\vec{b}}UP_{\vec{a}} U^\dag }\neq 0~~ \text{only if}~~|\vec{a}|=|\vec{b}|.
\end{eqnarray}
Let us define  the transition matrix $T_U$ as follows
\begin{eqnarray*}
T_U[\vec{b},\vec{a}]=\frac{1}{d^n}\Tr{P^\dag_{\vec{b}}UP_{\vec{a}} U^\dag },
\end{eqnarray*}
for any $\vec{a},\vec{b}\in V^n$. It is easy to see that $T_U$ is a unitary matrix. 
Here, due to the condition \eqref{con:eq_weight},  the unitary matrix can be decomposed as
 \begin{eqnarray*}
 T_U
 =\bigoplus^{n}_{k=0}
 T^{(k)}_U,
 \end{eqnarray*}
 where  $T^{(k)}_U$ is a $\binom{n}{k}(d^2-1)^k\times\binom{n}{k}(d^2-1)^k$ unitary matrix for any $0\leq k\leq n$, defined by
$T^{(k)}_U[\vec{b},\vec{a}]=\frac{1}{d^n}\Tr{P^\dag_{\vec{b}}UP_{\vec{a}} U^\dag }$ for any $\vec{a},\vec{b}$
with $|\vec{a}|=|\vec{b}|=k$.
Hence, 
\begin{eqnarray*}
\Tr{P^\dag_{\vec{b}}UOU^\dag}
=\sum_{\vec{a}:|\vec{a}|=|\vec{b}|}T^{|\vec{b}|}_U[\vec{b},\vec{a}] \Tr{OP_{\vec{a}}},
\end{eqnarray*}
and therefore, 
\begin{eqnarray*}
\sum_{\vec{b}:|\vec{b}|=k}P_{UOU^\dag}[\vec{b}]
=\sum_{\vec{b}:|\vec{b}|=k}P_{O}[\vec{b}],
\end{eqnarray*}
for any $0\leq k\leq n$. 
This implies that $I[UOU^\dag]=I[O]$.
Similarly, 
\begin{eqnarray*}
\sum_{\vec{b}:|\vec{b}|=k}P_{U^\dag OU}[\vec{b}]
=\sum_{\vec{b}:|\vec{b}|=k}P_{O}[\vec{b}],
\end{eqnarray*} 
and $I[U^\dag OU]=I[O]$. Therefore, $\CiS[U]=0$.

\noindent(2) This statement follows directly from the definition.

\noindent(3) Subadditivity under multiplication comes directly from the 
triangle inequality: 
\begin{eqnarray*}
\CiS[UV]\leq
\max_{O:\norm{O}_2=1}
\left|I[U VOV^\dag U^\dag]-I[VOV^\dag]\right|+
\max_{O:\norm{O}_2=1}
\left|I[V OV^\dag]-I[O]\right|.
\end{eqnarray*}
Hence, to prove the subadditivity under tensorization, we only need to prove that 
$\CiS[U\ot I]\leq \CiS[U]$. Let us assume that $U$ acts only on the $k$-qudit 
subsystem $S$ with $k\leq n$. 
Similarly to the proof of Theorem \ref{thm:main_inf},
let us define 
$O^{(1)}_{\vec{b}}$ on $(\complex^d)^{S^c}$ for any $\vec{b}\in V^{S}$ as \eqref{eq:O1} and, 
$A_{\vec{b}}=O^{(1)}_{\vec{c}}/\norm{O^{(1)}_{\vec{c}}}_2$. 
Define
$O^{(2)}_{\vec{c}}$ on $(\complex^d)^{S}$ for any $\vec{c}\in V^{S^c}$ as \eqref{eq:O2} and
$B_{\vec{c}}=O^{(2)}_{\vec{c}}/\norm{O^{(2)}_{\vec{c}}}_2$, so
$I[O]$  can be written as
\begin{eqnarray*}
I[O]=\sum_{\vec{c}\in V^{S^c}}
\norm{O^2_{\vec{c}}}^2_2
I[B_{\vec{c}}]
+\sum_{\vec{b}\in V^S}
\norm{O^1_{\vec{b}}}^2_2
I[A_{\vec{b}}].
\end{eqnarray*}
Similarly, $I[U\ot IO U^\dag \ot I]$ can be 
 written as 
 \begin{eqnarray*}
I [U\ot IO U^\dag \ot I]
 =\sum_{\vec{c}\in V^{S^c}}
\norm{O^{(2)}_{\vec{c}}}^2_2
I[UB_{\vec{c}}U^\dag]
+\sum_{\vec{b}\in V^S}
\norm{O^{(1)}_{\vec{b}}}^2_2
I[A_{\vec{b}}].
 \end{eqnarray*}
 Hence
 \begin{eqnarray*}
 |I [U\ot IO U^\dag \ot I]-I[O]|
 &\leq& \sum_{\vec{c}\in V^{S^c}}
\norm{O^{(2)}_{\vec{c}}}^2_2
|I[UB_{\vec{c}}U^\dag]-I[B_{\vec{c}}]|\\
&\leq& \CiS[U]\sum_{\vec{c}\in V^{S^c}}
\norm{O^{(2)}_{\vec{c}}}^2_2\\
&=& \CiS[U],
 \end{eqnarray*}
where we infer the second inequality  from the  definition of $\CiS$. The
last equality comes from the fact that $\sum_{\vec{c}\in V^{S^c}}
\norm{O^{(2)}_{\vec{c}}}^2_2=1$.
\end{proof}

We give two examples of stable unitaries.  In fact, all stable unitaries can be generated by these two types of unitaries.
\begin{enumerate}
\item 
A Kronecker product of single-qudit unitaries, $\bigotimes^n_{i=1} U_i$.
\item Swap gates, i.e.~the unitary mapping $\ket{\psi}\ket{\phi} \mapsto \ket{\phi}\ket{\psi}$.

\end{enumerate}

\begin{prop}
The set of stable unitaries is generated by 
 the single-qudit unitaries  and the swap unitaries.
\end{prop}
\begin{proof}
Given an $n$-qudit stable unitary $U$,
let us consider its action on $X_1$, where $X_i$ denotes the Pauli operator $X$ acting on the 
$i$-th qudit.
Since $U$ has zero circuit sensitivity, we have 
\begin{eqnarray*}
UX_1U^\dag
=\sum_{i\in A}\alpha_i
Q^X_i+\sum_{i\in B}\beta_i
Q^Y_i+\sum_{i\in C}\gamma_i
Q^Z_i,
\end{eqnarray*}
where $Q^X_i$ is written as 
$Q^X_i=\sum^{d-1}_{j=1} c_{ij}X^j_i$ with at least one coefficient $c_{ij}\neq 0$, 
and $A$ is the set of all indices $i$ such that $\alpha_i\neq 0$. The quantities $Q^Z_i$ and $C$ are similarly defined. 
Moreover, $Q^Y_i$ is defined as $Q^Y_i=\sum^{d-1}_{j,k=1}c_{ijk}X^j_iZ^k_i$ with at least one coefficient $c_{ijk}\neq 0$, 
and $B$ is the set of all indices $i$ for which $\beta_i\neq 0$.
Since $(UX_1U^\dag)^2 =I$, we have 
$|A|\leq 1$, $|B|\leq 1$ and $|C|\leq 1$. The first inequality holds because
if $|A|\geq 2$, then there exists 
two indices $i\neq j$ such that $(UX_1U^\dag)^2$ must contain some
term $Q^X_i\ot Q^X_j$, which contradicts with the fact that $(UX_1U^\dag)^2 =I$. 
Hence, we can simplify 
$UX_1U^\dag$ as
\begin{eqnarray*}
UX_1U^\dag
=\alpha_i
Q^X_i+\beta_j
Q^Y_j+\gamma_k
Q^Z_k.
\end{eqnarray*}
Since $(UX_1U^\dag)^2 =I$,
we have $i=j=k$. This holds because if $j\neq i$,
then $(UX_1U^\dag)^2 $ must contain the term 
$Q^X_i\ot Q^Y_j$. Hence,
we have 
\begin{eqnarray*}
UX_1U^\dag
=\alpha_i
Q^X_i+\beta_i
Q^Y_i+\gamma_i
Q^Z_i.
\end{eqnarray*}

Similarly, we have 
\begin{eqnarray*}
UZ_1U^\dag
=\alpha_j
Q^X_j+\beta_i
Q^Y_j+\gamma_i
Q^Z_j.
\end{eqnarray*}
If $i\neq j$, then 
$[UX_1U^\dag,UZ_1U^\dag]=0$, that is,
 $[X_1, Z_1]=0$,  which is impossible. Therefore
$i=j$, i.e., there exists a local unitary $V$ such that  for any $d\times d$ matrix $A$, 
$UA_1\ot I_{n-1}U^\dag=A'_i\ot I_{n-1}=VA_iV^\dag I_{n-1}$. Hence
\begin{eqnarray*}
V^\dag_1\SWAP_{1i}U=I_1\ot V_2,
\end{eqnarray*}
where $\SWAP_{1i}$ is the swap unitary between $1$ and $i$, and  $V_2$ has zero circuit sensitivity on 
$n-1$ qudits. By repeating the above process, we get that $U$ can be generated by the local unitaries and swap unitaries. 
\end{proof}

Stable unitaries also preserve  multipartite entanglement, where the entanglement is quantified by the  average R\'enyi-2 entanglement entropy:
\begin{eqnarray*}
\bar{S}^{(2)}(\rho)=\mathbb{E}S^{(2)}(\rho_A),
\end{eqnarray*}
where $\mathbb{E}:=\frac{1}{2^n}\sum_{A\subset [n]}$ denotes the expectation over subsets $A\subset [n]$;
$S^{(2)}(\rho_A)=-\log\Tr{\rho^2_A}$ denotes the R\'enyi-2 entanglement entropy; and $\rho_A$ denotes the reduced state of $\rho$ on the subset $A$.

\begin{cor}
Stable unitaries cannot increase the entanglement measure $\bar{S}^{(2)}$.
\end{cor}

\begin{proof}
It is easy to verify that 
both local unitaries and swap unitaries will not change 
this average entanglement R\'enyi-2 entropy, so the corollary follows from the proposition.
\end{proof}

This shows that sensitivity is necessary for a quantum computational advantage.

\begin{cor}
Given an $n$-qudit product state $\bigotimes^n_{i=1} \rho_i$ as input, a stable quantum circuit $U$, and a single-qudit measurement set $\set{N,I-N}$, the outcome probability can be classically simulated in $\poly(n,d)$ time.  
\end{cor}
\begin{proof}
Since the stable quantum circuit can be generated by local unitaries and swap gates, such quantum circuits with product input states and local measurements can be simulated efficiently on a classical computer.  
\end{proof}

\subsection{Matchgates are Gaussian stable gates}
In this section, we define variants of influence and circuit sensitivity, called Gaussian influence and Gaussian circuit sensitivity and show that matchgates have vanishing circuit sensitivity. We will show that Gaussian circuit sensitivity is necessary for a quantum computational advantage and that it provides
a good measure to quantify the non-Gaussianity of quantum circuits. 
Let us consider the influence based on the generators of a Clifford algebra for an $n$-qubit system. First, we introduce  $2n$ Hermitian operators $\gamma_i$ which satisfy the Clifford algebra relations
\begin{eqnarray}
\set{\gamma_i,\gamma_j}=2\delta_{i,j}I,~~~~\forall\, i,j=1,\ldots,2n.
\end{eqnarray}
Any linear operator can be
expressed as a polynomial 
of degree at most $2n$ as follows
\begin{eqnarray}
O=\sum_{S\subset[2n]} O_{S}\gamma^S,
\end{eqnarray}
where $\gamma^S=\prod_{i\in S} \gamma_i$. Then
\begin{eqnarray}
\mathop{\mathbb{E}}_{S\sim U}\left|\Tr{(\gamma^S)^\dag O}\right|^2=\norm{O}^2
_2.
\end{eqnarray}
Here the $\mathbb{E}_{S\sim U}$ denotes the expectation taken over all 
$S\subset [2n]$ with respect to the uniform distribution, that is, 
$\mathbb{E}_{S\sim U}=\frac{1}{2^{2n}}\sum_{S\subset [2n]}$. Hence, $\norm{O}_2=1$ leads to a probability 
distribution over $S$, which is defined as follows
\begin{eqnarray}\label{eq:proA_second}
P^G_O[S]=\frac{1}{2^{2n}}
\left|\Tr{(\gamma^S)^\dag O}\right|^2,\quad \forall  S\subset [2n].
\end{eqnarray}

Matchgates are an important family of tractable circuits, first proposed by 
Valiant in the context of counting problems \cite{valiant2002quantum}. Later, they were generalized to free fermionic quantum circuits, which are generated by a quadratic Hamiltonian in terms of Clifford generators 
$\set{\gamma_i}$, i.e., $H=\i\sum_{i,j}h_{ij}\gamma_i\gamma_j$ \cite{jozsa2008matchgates}. One important fact concerning matchgates (which 
are also called Gaussian gates) is that for each generator $\gamma_i$, $U\gamma_i U^\dag$ and $U^\dag \gamma_i U$ can always be written as
 linear combinations of $\gamma_i$ \cite{jozsa2008matchgates}. 
 
 We  provide a new interpretation of 
 matchgates via sensitivity, showing that they are the only unitaries which cannot change the 
 influence. To obtain this result, 
define the influence with respect to the generators of the Clifford algebra $\set{\gamma_i}_i$; we call this the Gaussian influence, to distinguish it from the previous definition. 
\begin{defn}[\textbf{Gaussian Influence}]
Given a linear operator $O$, 
the local influence at the 
$j$-th qudit is 
\begin{eqnarray}
I^G_j[O]
=\sum_{S: j\in S}
P^G_O[S],
\end{eqnarray}
and the total influence is the 
sum of all the local influences,
\begin{eqnarray}\label{eq:InG}
I^G[O]=
\sum_{j\in [2n]}
I^G_j[O]=\sum_{S\subset [2n]}
|S|P^G_O[S].
\end{eqnarray}
\end{defn}

Consider the Markov semigroup 
$P_t$ introduced by Carlen and Lieb in~\cite{carlen1993optimal},
\begin{eqnarray}
P_t(\gamma^S)
=\e^{-t|S|}\gamma^S\;.
\end{eqnarray}
Given an operator
$O$, we have 
\begin{eqnarray}
\frac{\partial}{\partial t}\norm{P_t(O)}^2_2\Big|_{t=0}
=-\sum_{S\subset [2n]}|S|
P^G_O[S]
=-I^G[O].
\end{eqnarray}

\begin{rema}
There is no obvious relationship between the Pauli weight and the Gaussian weight. In particular, there exist operators whose Pauli weight is 1 and Gaussian weight is $n$, and also operators whose Gaussian weight is 1 and Pauli weight is $n$. Consequently, there is no obvious relationship between the total influence $I$ and the total Gaussian influence $I^G$.

\end{rema}

Here, let us define the circuit sensitivity of a unitary with respect to $I^G$.
\begin{defn}
Given a unitary $U$, 
let us define the Gaussian circuit sensitivity $\CiS^G$ as the change of influence caused by the unitary evolution,
\begin{eqnarray}
\CiS^G[U]=
\max_{O:\norm{O}_2=1}
\left|I^G[U OU^\dag]-I^G[O]\right|.
\end{eqnarray}
We say that $U$ is Gaussian stable if $\CiS^G[U]=0$.
\end{defn}

\begin{thm}\label{thm:cha_match}
The Gaussian circuit sensitivity of an $n$-qudit unitary $U$ satisfies the following three properties:

\begin{enumerate}
    \item 
The unitary $U$  is Gaussian stable if and only if   $U$ is a matchagte.
\item
$\CiS^G[V_2UV_1]=\CiS^G[U]$ for any unitary $V_1$ and
matchgate $V_2$.
\item
$\CiS^G$ is subadditive  under multiplication, 
\begin{eqnarray}
\CiS^G[UV]\leq \CiS^G[U]+\CiS^G[V].
\end{eqnarray}
\end{enumerate}

\end{thm}

\begin{proof} \hfill

\noindent(1) On one hand, if $\CiS^G[U]=0$, then for any generator $\gamma_i$,
\begin{eqnarray*}
I^G[U\gamma_iU^\dag]=I^G[\gamma_i]=1.
\end{eqnarray*}
Hence, $U\gamma_iU^\dag$  can be written as $\sum_{j}c_j\gamma_j$. Similarly, $U^\dag \gamma_i U$  can be written as a
linear combination of $\set{\gamma_j}_j$.

On the other hand, if for any generator $\gamma_i$, both $U\gamma_iU^\dag$ and $U^\dag \gamma_iU$ can be written as a
linear combination of $\set{\gamma_j}_j$, then 
$U\gamma^SU^\dag$  and $U\gamma^S U^\dag$  can be written as a linear combination 
of $\set{\gamma^{S'}:S'\subset[2n], |S'|\leq |S|}$. Hence,  we have 
\begin{eqnarray}\label{con:eq_weight2}
\Tr{(\gamma^{S'})^\dag U\gamma^S U^\dag }\neq 0~~ \text{only if}~~|S'|=|S|.
\end{eqnarray}
Let us define the transition matrix $T_U$ as follows
\begin{eqnarray*}
T_U[S_1,S_2]=\frac{1}{2^n}\Tr{(\gamma^{S_1})^\dag U\gamma^{S_2} U^\dag },
\end{eqnarray*}
for any $S_1,S_2\subset [2n]$. It is easy to see that $T_U$ is a unitary matrix. 
Here, due to condition \eqref{con:eq_weight2},  the unitary matrix can be decomposed as
 \begin{eqnarray*}
 T_U
 =\bigoplus^{2n}_{k=0}
 T^{(k)}_U,
 \end{eqnarray*}
 where  $T^{(k)}_U$ is a $\binom{2n}{k}\times\binom{2n}{k}$ unitary matrix for any $0\leq k\leq 2n$, and defined as 
$T^{(k)}_U[S_1,S_2]=\frac{1}{2^n}\Tr{(\gamma^{S_1})^\dag U\gamma^{S_2} U^\dag  }$ for any $S_1,S_2$
with $|S_1|=|S_2|=k$.
Hence, 
\begin{eqnarray*}
\Tr{(\gamma^{S_1})^\dag UOU^\dag}
=\sum_{S_2:|S_2|=|S_1|}T^{(|S_1|)}_U[S_1,S_2] \Tr{(\gamma^{S_2})^\dag O},
\end{eqnarray*}
and therefore, 
\begin{eqnarray*}
\sum_{S_1:|S_1|=k}P^G_{UOU^\dag}[S_1]
=\sum_{S_1:|S_1|=k}P^G_{O}[S_1],
\end{eqnarray*}
for any $0\leq k\leq n$. 
This implies that $I^G[UOU^\dag]=I^G[O]$.
Similarly, 
\begin{eqnarray*}
\sum_{S_1:|S_1|=k}P^G_{U^\dag OU}[S_1]
=\sum_{S_1:|S_1|=k}P^G_{O}[S_1],
\end{eqnarray*} 
and we have   $I^G[U^\dag OU]=I^G[O]$. Therefore, $\CiS^G[U]=0$.

\noindent(2) It follows directly from the definition.

\noindent(3) It follows directly from the triangle inequality.
\end{proof}

Since matchgates can be simulated efficiently on a classical computer, the Gaussian stable gates cannot yield a quantum  advantage. From this we infer that  
Gaussian sensitivity is necessary for a quantum computational advantage. 
Since matchgates are sometimes called Gaussian operations and 
Gaussian circuit sensitivity can be used as a measure to quantify the non-Gaussian nature of quantum circuits, the Gaussian circuit sensitivity shows how "non-matchgate" a circuit is. 

The set of stable gates (i.e., $\CiS=0$) and Gaussian stable gates (i.e., $\CiS^G=0$) are  quite different.
For example, for an $n$-qubit system, the SWAP gates are stable but not Gaussian stable; on the other hand, the nearest neighbor (n.n.) $G(Z,X)$ gate is Gaussian stable, but not stable. Here the gate $G(Z,X)$ is defined as 
\begin{equation*}
 G(Z,X)= \left[ \begin{array}{cccc}
         1&0 &0 &0 \\
         0&0 &1 & 0\\
         0&1 &0 & 0\\
         0&0 &0 &-1 \\
    \end{array}
    \right].
\end{equation*}
Complementing this, we remark that a single-qubit unitary acting on the first qubit $U_1$ lies in the overlap of the two sets. We illustrate this in Figure \ref{Fig:Comparsion_stable gates}.

\definecolor{my_purple}{RGB}{155, 70, 212}
\begin{figure}[!ht]
\centering
\begin{tikzpicture}
 
\draw[
    color=red,
    very thick,
    xshift=1.5cm,
    yshift=-1.0675cm,
    rotate =40] (0,0) ellipse (3cm and 1.5cm);
 
\draw[
    color=blue,
    very thick,
    xshift=-1.5cm,
    yshift=-1.0075cm,
    rotate =-40] (0,0) ellipse (3cm and 1.5cm);
    
\node[] at (2,-0.5) {\color{red} n.n. $G(Z,X)$};
\node[] at (-2,-0.5) {\color{blue} $\SWAP$};
\node[] at (3,1.5) {\color{red} Gaussian stable unitaries};
\node[] at (-3,1.61) {\color{blue} stable unitaries};
\node[] at (0,-2) {\color{my_purple} $U_1$};
 
\end{tikzpicture}

\caption{A Venn diagram illustrating the overlap between the stable gate set and
    the Gaussian stable gate set on $n$-qubit systems, as explained in the text.}
    \label{Fig:Comparsion_stable gates}
\end{figure}
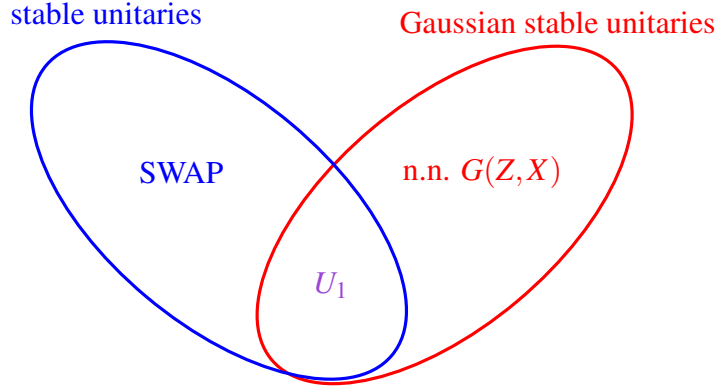

\begin{rema}\label{rema:tradeoff}
Here we consider the sensitivity of quantum circuits with respect to noise, where we define the stable gates (or circuits) as the gates with zero sensitivity.
 The  circuit sensitivity (or influence) may be used to quantify the
 classical simulation time, a question we plan to study in the future. 
 
 In classical computation, 
algorithmic stability is one of the fundamental properties of a classical algorithm, and it plays an important role in computational learning theory.  For example, it gives insight into the differential privacy of randomized algorithms~\cite{dwork2014algorithmic,dwork2016calibrating}, into the generalization error of learning algorithms~\cite{bousquet2002stability,bousquet2020sharper}, and so on. This implies that algorithmic stability is useful to understand learning.
Hence, one defines quantum algorithmic 
stability via influence (or circuit sensitivity) for quantum algorithms or circuits as a generalization of the classical theory. One can then study its application in quantum differential privacy~\cite{zhou2017differential,aaronson2019gentle} and in understanding the generalization error of quantum machine learning~\cite{banchi2021generalization,caro2021generalization,bu2021statistical,bu2021effects,bu2021rademacher,caro2022out,gibbs2022dynamical}. Besides, the stable gates (or circuits)
can be efficiently simulated on a classical computer, which shows that stability may not imply a quantum speedup. 

In summary, there appears to be 
a trade-off between  quantum computational speedup and the capability of generalization in quantum machine learning.
 \end{rema}

\subsection{Quantifying scrambling by influence on average case}
 
Here we clarify the relationship between influence and scrambling. Information scrambling measures the delocalization of quantum information by chaotic evolution. Scrambling prevents one from determining the initial conditions that precede  chaotic evolution through the use of local measurements. 
One well-known measure of scrambling is
the out-of-time-ordered commutator (OTOC). This  is defined as the Hilbert-Schmidt norm of the commutator between two initially commuting local Pauli strings after one operator evolves under the action of a unitary. Scrambling refers to the speed of growth of the OTOC. 
Mathematically, the OTOC is defined by 
\begin{equation}\label{Eq:Commutator}
C(t)=\frac{1}{2}\left\|[O_D(t),O_A]\right\|^2_{2}=1-\langle O_D(t)O_AO_D(t)O_A\rangle\;,
\end{equation}
where $O_D(t):=U_tO_DU^\dag_t$, the expectation value $\langle \bm{\cdot}\rangle$ is taken with respect to the $n$-qubit maximally mixed state $\mathbb{I}/d^n$, and 
$A, D$ denote two disjoint subregions of the $n$-qudit system. 
For simplicity, we take the local dimension to be $d=2$, i.e., the systems we consider are qubit systems.

If we restrict the regions $A$ and $D$ to be 1-qubit systems, then the average OTOC over all possible positions for $A$  can be expressed in terms of the influence of $O_D(t)$.  Without loss of generality, let us assume that the region $D$ is taken to be the $n$-th qubit. 
\begin{prop}[\textbf{Average OTOC-Influence Relation}]\label{prop:avOt_in}
If the region $D$ is the $n$-th qubit, then 
\begin{eqnarray}
\nonumber\mathbb{E}_A\mathbb{E}_{O_A}
\langle O_D(t)O_AO_D(t)O_A\rangle
=1-\frac{d^2}{d^2-1}\frac{1}{n-1}\sum^{n-1}_{j=1}
I_j[O_D(t)],\\
\end{eqnarray}
where $\mathbb{E}_A$ denotes the average over all positions $j\in [n-1]$ such that $O_A$ initially commutes with $O_D$, and 
$\mathbb{E}_{O_A}$ denotes the average over all local non-identity Pauli operators on position $A$. 
\end{prop}

\begin{proof}
Since $\langle O_D(t)O_AO_D(t)O_A\rangle$ can be written as the linear combination of the 
terms $\langle P_{\vec{a}}P_{c_j}P_{\vec{b}}P_{c_j}\rangle$ with $\vec{a},\vec{b}\in V^n$ and 
$P_{c_j}$ being the local non-identity Pauli operator on the $j$-th qubit, 
we first consider the average of $\langle P_{\vec{a}}P_{c_j}P_{\vec{b}}P_{c_j}\rangle$ with $P_{c_j}$  taking on all 
non-identity Pauli operators uniformly, 
\begin{eqnarray*}
\mathbb{E}_{P_j}\langle P_{\vec{a}}P_jP_{\vec{b}}P_j\rangle
&=&\frac{1}{d^2-1}\sum_{j=(s,t)\in V\backslash(0,0)}\langle P_{\vec{a}}P_jP_{\vec{b}}P_j\rangle\\
&=&\delta_{\vec{a},\vec{b}}
\frac{1}{d^2-1}(d^2\delta_{a_j,0}-1)\\
&=&\delta_{\vec{a},\vec{b}}
\left(1-|a_j|\frac{d^2}{d^2-1}\right).
\end{eqnarray*}
Hence,  any $O_D(t)$ can be written as $O_D(t)=\sum_{\vec{a}}\frac{1}{d^n}\Tr{P_{\vec{a}}O_D(t)}P_{\vec{a}}$, 
\begin{eqnarray*}
\mathbb{E}_A\mathbb{E}_{O_A}
\langle O_D(t)O_AO_D(t)O_A\rangle
&=&\frac{1}{n-1}\sum^{n-1}_{j=1}
\sum_{\vec{a}\in V^n} \left(1-|a_j|\frac{d^2}{d^2-1}\right)
P_{O_D}[\vec{a}]\\
&=&
1-\frac{d^2}{d^2-1}\cdot\frac{1}{n-1}
\sum_{\vec{a}\in V^n}\left[\sum^{n-1}_{j=1}|a_j|\right]
P_{O_D}[\vec{a}]\\
&=&1-\frac{d^2}{d^2-1}\cdot\frac{1}{n-1}\sum^{n-1}_{j=1}
\sum_{\vec{a}\in V^n}|a_j|
P_{O_D}[\vec{a}]\\
&=&1-\frac{d^2}{d^2-1}\frac{1}{n-1}\sum^{n-1}_{j=1}
I_j[O_D(t)],
\end{eqnarray*}
where $I_j$ is defined as 
$I_j[O]=\sum_{a_j\neq 0}P_O[\vec{a}]$.
\end{proof}

Proposition \ref{prop:avOt_in} ensures that the average OTOC tends to
\begin{eqnarray*}
1-\frac{d^2}{d^2-1}\frac{1}{n-1}\sum^{n-1}_{j=1}
I_j[O_D(t)]
\to 1-\frac{d^2}{d^2-1}\frac{1}{n}
I[O_D(t)] ~~\text{as}~~ n \to \infty.
\end{eqnarray*}
This provides the relations between scrambling and the total influence.
Aside from the OTOC, higher-order OTOCs, such as the 8-point correlator, can also be 
related to the total influence on average (See Appendix~\ref{Appen:OTOC}).

\section{Quantum Fourier entropy and influence }\label{sec:FEn_inf}
Here, we define the quantum Fourier entropy $H[O]$ and show its relationship with the influence $I[O]$. We shall show that the quantum Fourier entropy can be
used as a measure of magic in quantum circuits, which we call the ``magic entropy''. In addition, we use results on 
quantum Fourier entropy and influence 
to obtain the relations between  magic and sensitivity (or Gaussian sensitivity).

\subsection{Quantum Fourier entropy-influence relation and conjecture}
\label{subsec:FeI_inf}
\begin{defn}[\textbf{Quantum Fourier Entropy and Min-entropy}]
Given a linear $n$-qudit operator $O$ with $\norm{O}_2=1$, the quantum Fourier entropy $H[O]$
is
\begin{eqnarray}
H[O]=H[P_O]
=-\sum_{\vec{a}\in V^n}P_O[\vec{a}]\log P_O[\vec{a}],
\end{eqnarray}
with $\set{P_O[\vec{a}]}$ being the probability distribution defined in \eqref{eq:proA}. 
The quantum Fourier min-entropy $H_{\infty}[O]$ is
\begin{eqnarray}
H_{\infty}[O]=H_{\infty}[P_O]
=\min_{\vec{a}\in V^n}\log \frac{1}{P_O[\vec{a}]}.
\end{eqnarray}
\end{defn}
One can also define the quantum Fourier R\'enyi entropy as
\begin{eqnarray}
H_{\alpha}[O]
=H_{\alpha}[P_O]
=\frac{1}{1-\alpha}\log
\left(
\sum_{\vec{a}}
P^{\alpha}_O[\vec{a}]
\right).
\end{eqnarray}

In the study of classical Boolean functions, Friedgut and Kalai proposed the now well-known Fourier entropy-influence conjecture~\cite{friedgut1996every}. 
Another well-known, but weak, conjecture is the Fourier min-entropy-influence conjecture. Appendix~\ref{Appen:FourierAnalysisClassical} provides a brief introduction to the Fourier entropy-influence conjecture for Boolean functions.

\begin{thm}[Weak QFEI]\label{thm:main1}
For any linear operator $O$ on an $n$-qudit system with $\norm{O}_2=1$, we have 
\begin{eqnarray}
H[O]\leq c[\log n+\log d]I[O]+h[P_O[\vec{0}]],
\end{eqnarray}
where $h(x):=-x\log x-(1-x)\log(1-x)$ is the binary entropy and $c$ is a universal constant. 
Here, $c$ can be taken to be $2$.
\end{thm}
\begin{proof}
Let us define  a new probability distribution $\set{W_k[O]}_k$
on the set $[n] $ as follows
\begin{eqnarray*}
W_k[O]
=\sum_{\vec{a}\in V^n:|\vec{a}|=k}
P_O[\vec{a}].
\end{eqnarray*}
Therefore, the total influence $I[O]$ can be rewritten as 
\begin{eqnarray*}
I[O]
=\sum_{\vec{a}\in V^n}
|\vec{a}|P_O[\vec{a}]
=\sum_{k} kW_k[O]. 
\end{eqnarray*}
Hence, the quantum Fourier entropy can be written as
\begin{eqnarray*}
H[O]&=&
\sum_{\vec{a}\in V^n}
P_O[\vec{a}]
\log\frac{1}{P_O[\vec{a}]}\\
&=&\sum_{\vec{a}\in V^n}
P_O[\vec{a}]
\left(\log\frac{W_{|\vec{a}|}[O]}{P_O[\vec{a}]}+
\log \frac{1}{W_{|\vec{a}|}[O]}\right)
\\
&=&\sum_{\vec{a}\in V^n}
P_O[\vec{a}]
\log\frac{W_{|\vec{a}|}[O]}{P_O[\vec{a}]}+\sum_{\vec{a}\in V^n}P_O[\vec{a}]
\log \frac{1}{W_{|\vec{a}|}[O]}\\
&=&
\sum_{k}
W_k[O]
\sum_{\vec{a}:|\vec{a}|=k}
\frac{P_O[\vec{a}]}{W_k[O]}\log\frac{W_k[O]}{P_O[\vec{a}]}
+\sum_{k}
W_k[O]\log\frac{1}{W_k[O]}.
\end{eqnarray*}
Note that if $W_k[O]\neq 0$, then $\frac{P_O[\vec{a}]}{W_k[O]}$ is a
probability distribution on the set $S_k=\set{\vec{a}\in V^n: |\vec{a}|=k}$. Hence
\begin{eqnarray*}
\sum_{\vec{a}:|\vec{a}|=k}
\frac{P_O[\vec{a}]}{W_k[O]}\log\frac{W_k[O]}{P_O[\vec{a}]}
&\leq& \log |S_k|
\leq \log\left(\binom{n}{k}(d^2-1)^k\right)\\
&\leq& k(\log n+\log(d^2-1)).
\end{eqnarray*}
Therefore, we have 
\begin{eqnarray*}
\sum_{k}
W_k[O]
\sum_{\vec{a}:|\vec{a}|=k}
\frac{P_O[\vec{a}]}{W_k[O]}\log\frac{W_k[O]}{P_O[\vec{a}]}
&\leq& \sum_{k}
W_k[O]k(\log n+\log(d^2-1))\\
&=&[\log n+\log(d^2-1)]
I[O].
\end{eqnarray*}

Next, let us prove that
$\sum_{k}
W_k[O]\log\frac{1}{W_k[O]}\leq I[O]+h(P_O[\vec{0}])$. 
First, 
if $\Tr{O}=0$, then 
$
H[O]\leq I[O]
$.
This comes from the positivity 
of the relative entropy between the probability distributions
$\vec{W}=\set{W_k[O]}_k$
 and $\vec{p}=\set{p_k}_k$, with $p_k=2^{-k}$ for $1\leq k\leq n$
 and 
$p_0=2^{-n}$, which can be
expressed as 
\begin{eqnarray*}
D(\vec{W}\|\vec{p})=
\sum^n_{k=0}W_k[O]
\log \frac{W_k[O]}{p_k[O]}
=\sum_k kW_k[O]
+\sum_kW_k[O]\log W_k[O]
\geq 0.
\end{eqnarray*}
If $\Tr{O}\neq 0$, 
let us
us define a new operator
\begin{eqnarray*}
O'=\frac{1}{1-W_0[O]}\sum_{\vec{a}\neq 0}
O_{\vec{a}}P_{\vec{a}}.
\end{eqnarray*}
Then for this new operator $O'$, 
we have 
\begin{eqnarray*}
H[O']\leq I[O'],
\end{eqnarray*}
and 
\begin{eqnarray*}
I[O]=(1-W_0[O])
I[O'].
\end{eqnarray*}
Hence, 
\begin{eqnarray*}
&&\sum_kW_k[O]
\log\frac{1}{W_k[O]}\\
&=&W_0[O]\log \frac{1}{W_0[O]}
+
\sum_{k\geq 1}W_k[O]
\log\frac{1}{W_k[O]}\\
&=&W_0[O]\log\frac{1}{ W_0[O]}
+
\sum_{k\geq 1}(1-W_0[O])W_k[O']
\log\frac{1}{(1-W_0[O])W_k[O']}\\
&=&W_0[O]\log \frac{1}{W_0[O]}
+(1-W_0[O])\log \frac{1}{(1-W_0[O])}\\
&&\qquad\qquad\qquad+
(1-W_0[O])\sum_{k\geq 1}W_k[O']
\log\frac{1}{W_k[O']}\\
&\leq &h(W_0[O])
+(1-W_0)I[O']\\
&=&I[O]+h(P_O[\vec{0}]),
\end{eqnarray*}
where $h$ denotes the binary entropy $h(x)=-x\log x-(1-x)\log (1-x)$.
This completes the proof of the theorem.
\end{proof}

Now, let us consider the 
quantum Fourier entropy-influence conjecture on qubit systems, which improves upon Theorem \ref{thm:main1}.

\begin{con}[\textbf{Quantum Fourier Entropy-Influence Conjecture}]
Given a Hermitian operator $O$ on $n$-qubit systems with $O^2=I$,
\begin{eqnarray}
H[O]\leq c I[O],
\end{eqnarray}
where the constant $c$ is independent of $n$.
\end{con}

\begin{prop}[\textbf{QFEI Implies FEI}]\label{prop:QFEi_imp}
If QFEI is true for  Hermitian operators $O$ on $n$-qubit system with $O^2=I$, then FEI is also true.
\end{prop}
\begin{proof}
Consider any function $f:\set{-1,1}^n\to\set{-1,1}$ with the corresponding  Fourier expansion 
$f(x)=\sum_{S\subset [n]}\hat{f}(S)x_S$. 
Let us define the following observable
\begin{eqnarray*}
O_f=\sum_{S\subset[n]}\hat{f}(S)X^S.
\end{eqnarray*}
where $X^S:=\prod_{i\in S} X_i$ and $X_i$ is the Pauli X operator on the $i$-th qubit.
$O_f$ is a Hermitian operator with $O^2_f=I$.
Note that 
\begin{eqnarray*}
O_f\ket{x}=f(x)\ket{x}, ~~\forall x\in \set{-1,1}^n,
\end{eqnarray*}
where  $\ket{\pm1}=\frac{1}{\sqrt{2}}(\ket{0}+\ket{1})$ and $\ket{0},\ket{1}$ are the eigenstates of 
the Pauli $Z$ operator. 
Hence 
$
\inner{O_f}{P_{\vec{a}}}=\hat{f}(S)
$ when $P_{\vec{a}}=X^S$, and 
$\inner{O_f}{P_{\vec{a}}}=0$ otherwise. 
That is, 
\begin{eqnarray*}
H[f]&=&H[O_f], \\
I[f]&=&I[O_f].
\end{eqnarray*}
This completes the proof of the proposition.
\end{proof}
Similarly, QFMEI is a quantum generalization of FMEI. 
\begin{prop}[\textbf{QFMEI Implies FMEI}]
If QFMEI is true for all quantum Boolean functions, FMEI is also true.
\end{prop}
\begin{proof}
The proof follows the same lines as the proof of Proposition \ref{prop:QFEi_imp}. 
\end{proof}

\subsection{Magic entropy-circuit sensitivity relation}
Magic is an important resource in quantum computation, as a quantum circuit without magic provides no quantum advantage. The Gottesman-Knill theorem states that Clifford unitaries with stabilizer states and Pauli measurements can be efficiently simulated on a classical computer \cite{gottesman1998heisenberg,aaronson2004improved}. Here, a Clifford unitary 
is defined as a unitary which maps a Pauli operator to a Pauli operator. Since any Pauli operator is generated by  the product of weight-1 Pauli operators, the
Clifford unitaries are precisely those unitaries which map any weight-1 Pauli operator to a Pauli operator. 
For a non-Clifford unitary, an important task is to quantify the amount of magic in the unitary. 
Here, we introduce a new concept, which we call the magic entropy. 

\begin{defn}[\textbf{Magic Entropy}]
Given a unitary $U$, the magic entropy $M[U]$ is 
\begin{eqnarray}
M[U]:=
\max_{\textnormal{$O$: weight-1 Pauli}}H[UOU^\dag].
\end{eqnarray}

\end{defn}
Since the quantum Fourier entropy of any weight-1 Pauli is always 0, the magic entropy can also be written as
follows
\begin{eqnarray}
M[U]=
\max_{\textnormal{$O$: weight-1 Pauli}}(H[UOU^\dag]-H[O]),
\end{eqnarray}
which also quantifies the change of quantum Fourier entropy on 
weight-1 Pauli operators. 

\begin{prop}
The magic entropy $M[U]$ satisfies the following three properties:
\begin{enumerate}
    \item Faithfulness: $M[U]\geq 0$, and $M[U]=0$ if and only if $U$ is a Clifford unitary.
\item
Invariance under multiplication by Clifford unitaries: 
$M[VU]=M[U]$ for any 
  Clifford unitary $V$.
\item
Maximization under tensorization:
$M[U_1\ot U_2]=\max\set{M[U_1],M[U_2]}$ for 
any unitaries $U_1$ and $U_2$.
\end{enumerate}
\end{prop}

\begin{proof}
These properties follow directly from the definition of magic entropy.
\end{proof}

\begin{example}
Let us consider a widely-used single-qubit non-Clifford  $T$ gate, which is defined as 
$T=\left[
\begin{array}{cc}
1&0\\
0&\e^{\i\pi/4}
\end{array}
\right]$. The magic entropy of $T$  is $M[T]=1$.
\end{example}

Based on the relations between quantum Fourier entropy and influence in \S \ref{subsec:FeI_inf}, we can obtain the connection between 
magic entropy and circuit sensitivity.
\begin{prop}[\textbf{Magic-Sensitivity Relation}]\label{prop:mag_inf}
Given an $n$-qudit unitary $U$, the magic entropy and circuit sensitivity satisfy the following relation:
\begin{eqnarray}
M[U]\leq c[\log n+\log d](\CiS[U]+1).
\end{eqnarray}
\end{prop}
\begin{proof}
Based on Theorem \ref{thm:main1}, we have 
\begin{eqnarray*}
H[UOU^\dag]\leq c[\log n+\log d]I[UOU^\dag].
\end{eqnarray*}
 Besides, as 
$I[O]=1$ for a weight-1 Pauli operator $O$, we have  
\begin{eqnarray*}
I[UOU^\dag]\leq \CiS[U]+I[O]=\CiS[U]+1.
\end{eqnarray*} 
Thus
\begin{eqnarray*}
H[UOU^\dag]\leq c[\log n+\log d](\CiS[U]+1),
\end{eqnarray*}
for any weight-1 Pauli operator $O$.
\end{proof}

\begin{prop}\label{prop:mag_inf_2}
If the  QFEI conjecture holds for an $n$-qubit system, then for any 
$n$-qubit unitary $U$, 
\begin{eqnarray}
M[U]\leq c(\CiS[U]+1).
\end{eqnarray}
\end{prop}

\begin{proof}
The proof is similar to that for Proposition \ref{prop:mag_inf}.
\end{proof}

Since the Gaussian influence $I^G$ has properties similar to the influence $I$,  we can get the 
following connection between quantum Fourier entropy and Gaussian influence by a similar proof, which we call the weak Quantum 
Fourier entropy-Gaussian influence relation (QFEGI). Hence, it also implies the connection between magic
entropy and Gaussian circuit sensitivity for quantum circuits.

\begin{thm}[\textbf{Weak QFEGI}]
For any linear operator $O$ on an $n$-qubit system with $\norm{O}_2=1$, we have 
\begin{eqnarray}
H[O]\leq c\log(2n)I^G[O]+h[P_O[\vec{0}]],
\end{eqnarray}
where $h(x):=-x\log x-(1-x)\log(1-x)$ is the binary entropy and $c$ is a universal constant. 
\end{thm}
\begin{proof}
The proof is similar to that of Theorem
\ref{thm:main1}, so we omit it here.
\end{proof}

\begin{prop}[\textbf{Magic-Gaussian Sensitivity Relation}]
Given an $n$-qudit unitary $U$,  the magic entropy and Gaussian circuit sensitivity satisfy the following relation
\begin{eqnarray}
M[U]\leq c(\log 2n)(\CiS^G[U]+1).
\end{eqnarray}
\end{prop}
\begin{proof} The proof is similar to that of Theorem
\ref{prop:mag_inf}, so  we omit it here.
\end{proof}

\section{Magic and circuit complexity}\label{sec:magic}

\subsection{A lower bound on circuit cost from magic-influence relation}\label{sec:mag_1}

As the influence of a unitary evolution can provide a lower bound on the circuit cost, the magic-influence relation directly implies a lower bound on the circuit 
cost by the amount of magic.

\begin{prop}
The circuit cost of a unitary $U\in \SU(d^n)$ satisfies the following lower bound given by the magic entropy

\begin{eqnarray}
\Cost(U)+1\geq \frac{1}{c_d\log n}M[U]
.
\end{eqnarray}
\end{prop}
\begin{proof}

This is because
\begin{eqnarray*}
M[U]
\leq c_d\log (n)
(\CiS[U]+1)
\leq  c_d\log(n) (\Cost(U)+1),
\end{eqnarray*}
where the first inequality comes from Proposition  \ref{prop:mag_inf}, and the second inequality comes from 
Theorem \ref{thm:cost_In}. 
\end{proof}

\begin{prop}
If the QFEI conjecture holds for $n$-qubit systems, then the circuit cost of a  unitary $U\in \SU(2^n)$ satisfies the following  
bound
\begin{eqnarray}
\Cost(U)+1\geq \frac{1}{c}M[U].
\end{eqnarray}
\end{prop}

\begin{proof}
This is because
\begin{eqnarray*}
M[U]
\leq c
(\CiS[U]+1)
\leq  c [\Cost(U)+1],
\end{eqnarray*}
where the first inequality comes from  Proposition~\ref{prop:mag_inf_2},  and the second inequality comes from 
Theorem~\ref{thm:cost_In}. 

\end{proof}

\subsection{A lower bound on circuit cost by magic power}\label{sec:mag_2}
In  subsection \S \ref{sec:mag_1}, we obtain a lower bound on the circuit cost based on the magic-influence relation. This lower bound has a $\log n$ factor, which can be removed under  the quantum Fourier entropy-influence conjecture. 
In this subsection, our goal is to get rid of the $\log n$ factor without the conjecture. First, let us introduce another concept called magic power, 
which is a generalization of magic entropy. 

\begin{defn}[\textbf{Magic Power}] Given a unitary $U$, 
the magic power $\mathcal{M}[U]$ is 
the maximal magic 
generated by $U$, 
\begin{eqnarray}
\mathcal{M}[U]=
\max_{O:\norm{O}_2=1}
\left|H[U OU^\dag]-H[O]\right|.
\label{eq:magic_power}
\end{eqnarray}
\label{def:magic_power}
\end{defn}
It is easy to see that the magic power satisfies $\mathcal{M}[U] \geq M[U]$. Let us first discuss some properties of the magic power. 
\begin{lem}
The magic power equals
\begin{eqnarray}
\mathcal{M}[U]=
\max_{O:\norm{O}_2=1,\Tr{O}=0}
\left|H[U OU^\dag]-H[O]\right|,
\end{eqnarray}
that is, the maximization is taken over all traceless operators with $\norm{O}_2=1$.
\end{lem}
\begin{proof}
Let us define a new operator $O'$:
\begin{eqnarray*}
O'=\frac{1}{\sqrt{1-P_O[\vec{0}]}}
\left(O-\frac{\Tr{O}}{d^n}I\right).
\end{eqnarray*}
If $\Tr{O}\neq 0$, then $O'$ satisfies the condition $\Tr{O'}=0$ and $\norm{O'}_2=1$. 
Since $P_{UOU^\dag}[\vec{0}]=P_{\vec{O}}[\vec{0}]$,  $H[UOU^\dag]$
 and $H[O]$ can be rewritten as
 \begin{eqnarray*}
 H[O]&=&h\left[P_{O}[\vec{0}]\right]
 +(1-P_O[\vec{O}])M[O'],\\
 H[UOU^\dag]&=&h\left[P_{O}[\vec{0}]\right]
 +(1-P_O[\vec{O}])M[UO'U^\dag].
 \end{eqnarray*}
Hence we have 
\begin{eqnarray*}
H[UOU^\dag]
 -H[O]
 =(1-P_{O}[\vec{0}])
 (H[UO'U^\dag]
 -H[O']).
\end{eqnarray*}
Therefore, the maximization is obtained from traceless operators. 
\end{proof}

\begin{prop}
The magic power $\mathcal{M}[U]$ satisfies the following three properties;

\begin{enumerate}
    \item 
Magic power is faithful: $\mathcal{M}[U]\geq 0$, and $\mathcal{M}[U]=0$ if and only if $U$ is a Clifford unitary.
\item 
Magic power is invariant under multiplication by Cliffords: 
$\mathcal{M}[V_2UV_1]=\mathcal{M}[U]$ for any 
unitary $V_1$ and Clifford unitary $V_2$.
\item
Magic power is subadditive under multiplication and tensorization:
\begin{eqnarray}
\mathcal{M}[UV]\leq \mathcal{M}[U]+\mathcal{M}[V],
\qquad\mathcal{M}[U\ot V]\leq \mathcal{M}[U]+\mathcal{M}[V].
\end{eqnarray}
\end{enumerate}
\end{prop}
\begin{proof} 
\hfill

\noindent(1) $\mathcal{M}[U]\geq 0$ comes  directly from the definition of $\mathcal{M}[U]$. 
If $\mathcal{M}[U]=0$, it implies
that $H[UP_{\vec{a}}U^\dag]=0$ for any Pauli operator $P_{\vec{a}}$, that is 
$UP_{\vec{a}}U^\dag$ is a Pauli operator.  Hence the unitary  $U$ is a
Clifford unitary. If $U$ is a Clifford unitary—i.e.~if $U$ always maps Pauli operators to Pauli operators—then the probability distribution $\set{P_{UOU^\dag}[\vec{a}]}$ is 
 equivalent to $\set{P_{O}[\vec{a}]}$ up to some permutation. Hence
$H[U OU^\dag]=H[O]$.

\noindent(2)
This follows directly from the definition of $\mathcal{M}[U]$. 

\noindent(3) Subadditivity under multiplication comes directly from the triangle inequality, that is 
\begin{eqnarray*}
\mathcal{M}[UV]&\leq&
\max_{O:\norm{O}_2=1}
\left|H[U VOV^\dag U^\dag]-H[VOV^\dag]\right|+
\max_{O:\norm{O}_2=1}
\left|H[V OV^\dag]-H[O]\right|\\
&=&\mathcal{M}[U]+\mathcal{M}[V].
\end{eqnarray*}
Hence, to prove subadditivity under tensorization, we only need to prove that 
$\mathcal{M}[U\ot I]\leq \mathcal{M}[U]$. Let us assume that $U$ acts on only a $k$-qudit 
subsystem $S$ with $k\leq n$. 
Let us define $O_{\vec{c}}$ on $(\complex^d)^{S}$ for any $\vec{c}\in V^{S^c}$ as follows
\begin{eqnarray}\label{eq:ooc}
O_{\vec{c}}
=\frac{1}{d^{n-k}}
\text{Tr}_{S^c}[O P_{\vec{c}}],
\end{eqnarray}
and it is easy to verify that 
$
\sum_{\vec{c}\in V^{S^c}}\norm{O_{\vec{c}}}^2_2=1.
$
Defining $B_{\vec{c}}=O_{\vec{c}}/\norm{O_{\vec{c}}}_2$, we get that
$H[O]$  can be written as
\begin{eqnarray*}
H[O]
=\sum_{\vec{c}}
\norm{O_{\vec{c}}}^2_2H[B_{\vec{c}}]
-\sum_{\vec{c}}
\norm{O_{\vec{c}}}^2_2\log\norm{O_{\vec{c}}}^2_2.
\end{eqnarray*}
Similarly, 
\begin{eqnarray*}
H[U\ot I O U^\dag \ot I]
=-\sum_{\vec{c}}
\norm{O_{\vec{c}}}^2_2 H[UB_{\vec{c}}U^\dag]
-\sum_{\vec{c}}
\norm{O_{\vec{c}}}^2_2\log\norm{O_{\vec{c}}}^2_2.
\end{eqnarray*}
Hence
\begin{eqnarray*}
|H[U\ot I O U^\dag \ot I]-H[O]|
\leq \sum_{\vec{c}}
\norm{O_{\vec{c}}}^2_2
\left|H[UB_{\vec{c}}U^\dag]- H[B_{\vec{c}}]\right|
\leq \mathcal{M}[U].
\end{eqnarray*}
Hence, we obtain the result.
\end{proof}
\begin{example}
By a simple calculation, the magic power of a $T$ gate is $\mathcal{M}[T]=1$. Moreover, for $n$ copies the $T$ gate, namely $T^{\ot n}$, its magic power is $\mathcal{M}[T^{\ot n}]=n$, whereas its magic entropy $M[T^{\ot n}]=1$,
which follows directly from the maximization of magic entropy under tensorization. 
This example illustrates that magic power may be much larger than magic entropy for the same unitary.
\end{example}

We now introduce the magic rate, which can be used to 
quantify small incremental magic for a given unitary evolution.
\begin{defn}[\textbf{Magic Rate}]
Given an $n$-qudit Hermitian 
 Hamiltonian $H$ and a linear operator $O$ with $\norm{O}_2=1$,  the magic rate of the
 unitary $U_t=\e^{-\i t H}$ acting on $O$ is
\begin{eqnarray}
R_M(H,O)
=\frac{\d}{\d t}H[U_tOU^\dag_t]\Big|_{t=0}.
\end{eqnarray}

 \end{defn}

First, let us provide an analytic
formula for the magic rate by a direct calculation as follows,
\begin{eqnarray*}
R_M(H,O)&=&\frac{\i}{d^{2n}}\sum_{\vec{a}\in V^n}
\left(\Tr{[O, H]P_{\vec{a}}}
\Tr{OP^\dag_{\vec{a}}}\log P_{O}[\vec{a}]\right.\\
&&\qquad\qquad\qquad\left.
+\Tr{[O, H]P^\dag_{\vec{a}}}
\Tr{OP_{\vec{a}}}\log P_{O}[\vec{a}]
\right).
\end{eqnarray*}

\begin{lem}
Consider the function $g(x)=x(\log x)^2$ with $x\in[0,1]$. Then   $0\leq g(x)\leq g(\e^{-2})=(2\log \e)^2/\e^2$ for $x\in [0,1]$. Moreover, $g(x)$ is increasing on $[0,\e^{-2}]$ and 
decreasing on $[\e^{-2},1]$. 

\end{lem}
\begin{proof}
This lemma follows from elementary calculus. See Fig.~\ref{Fig:g} for a plot of the function $g(x)$.
\end{proof}

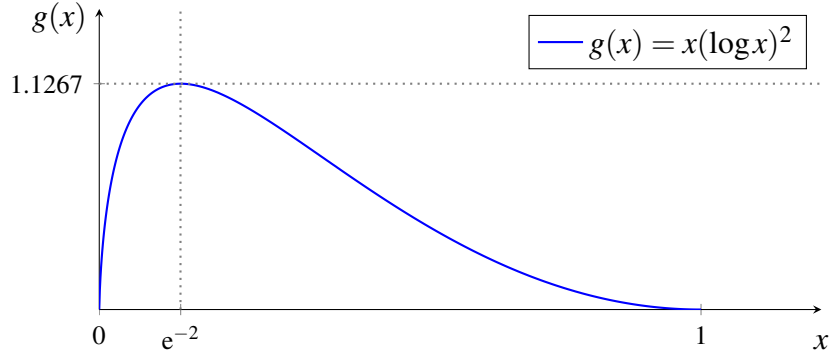
\begin{figure}[!ht]
    \centering
\begin{tikzpicture}
\begin{axis}[
    axis lines = left,
    xlabel = \(x\),
    ylabel = {\(g(x)\)},
    xmax = 1.2,
    ymax = 1.5,
    y=8cm/3,
    x=8cm,
    xtick={0,0.1353352832,1},
    xticklabels={\footnotesize 0, \footnotesize $\e^{-2}$, \footnotesize 1},
    ytick={1.126730642},
    yticklabels={\footnotesize 1.1267},
    x label style={at={(axis description cs:1,-0.05)},anchor=north},
    y label style={at={(axis description cs:-0.055,0.88)},rotate=270,anchor=south}
]

\addplot [
    domain=0:1, 
    samples=1000, 
    color=blue,
    thick
]
{x*log2(x)^2};
\addlegendentry{\(g(x) = x (\log x)^2\)}

\draw[thick,color=gray,dotted] (axis cs:0.135,\pgfkeysvalueof{/pgfplots/ymin}) 
-- 
(axis cs:0.135,\pgfkeysvalueof{/pgfplots/ymax});

\draw[thick,color=gray,dotted] (axis cs:\pgfkeysvalueof{/pgfplots/xmin},1.1267) 
-- 
(axis cs:\pgfkeysvalueof{/pgfplots/ymax},1.1267);

\end{axis}
\end{tikzpicture}
\caption{A plot of the function $g(x)=x(\log x)^2$ for $x\in[0,1]$, where the logarithm is taken to be of base 2. The maximum value of $g(x)$ is $g(\e^{-2}) \approx 1.1267$, which occurs at $x = \e^{-2} \approx 0.135$. The function $g(x)$ vanishes at both $x=0$ and $x=1$. In addition, it is increasing on $[0,\e^{-2}]$ and decreasing on $[\e^{-2},1]$.
}
\label{Fig:g}
\end{figure}

\begin{lem}\label{lem:magic_ra}
Given an $n$-qudit Hamiltonian $H$ and
a linear operator $O$ with $\norm{O}_2=1$,  we have
\begin{eqnarray}
|R_M(H, O)|\leq 
8d^{n}\norm{H}_{\infty}\log(\e)/\e.
\label{eq:RMHOL}
\end{eqnarray}
\end{lem}
\begin{proof}
The Schwarz inequality yields 
\begin{eqnarray*}
&&\frac{1}{d^{2n}}\sum_{\vec{a}\in V^n}
\left|\Tr{[O, H]P_{\vec{a}}}\right|
\left|\Tr{OP^\dag_{\vec{a}}}\log P_{O}[\vec{a}]\right|\\
&\leq& \norm{[H,O]}_2
\left(\sum_{\vec{a}\in V^n}
P_O[\vec{a}]\log^2 P_{O}[\vec{a}]\right)^{1/2}\\
&\leq&\norm{[H,O]}_2
(d^{2n}g(\e^{-1}))^{1/2}\\
&\leq &2d^n\norm{H}_{\infty}\sqrt{g(\e^{-2})},
\end{eqnarray*}
where  the 
second inequality come from the fact that $g(x)\leq g(\e^{-2})$ and the last 
inequality comes from the H\"older inequality.
Similarly, 
\begin{eqnarray}
\frac{1}{d^{2n}}\sum_{\vec{a}\in V^n}|\Tr{[O, H]P^\dag_{\vec{a}}}
\Tr{OP_{\vec{a}}}\log P_{O}[\vec{a}]|
\leq 2d^n\norm{H}_{\infty}\sqrt{g(\e^{-2})}.
\end{eqnarray}
Therefore, we get the bound in \eqref{eq:RMHOL}.
\end{proof}

\begin{thm}[\textbf{Small Incremental Magic}]\label{thm:mag_rate}
Given an $n$-qudit system with the Hamiltonian $H$ acting on a $k$-qudit subsystem, and a linear operator $O$
with $\norm{O}_2=1$, one has
\begin{eqnarray}\label{eq:exp_dep}
|R_M(H, O)|\leq 
8d^{k}\norm{H}_{\infty}\log(\e)/\e.
\end{eqnarray}
\end{thm}
\begin{proof}
Since $H$ acts on a $k$-qudit subsystem, there exists a 
subset $S$ of size $k$ such that $H=H_S\ot I_{S^c}$.
Define $O_{\vec{c}}$ on $(\complex^d)^{S}$ for  $\vec{c}\in V^{S^c}$ by
\begin{eqnarray*}
O_{\vec{c}}
=\frac{1}{d^{n-k}}
\text{Tr}_{S^c}[O P_{\vec{c}}].
\end{eqnarray*}
Note that 
$
\sum_{\vec{c}\in V^{S^c}}\norm{O_{\vec{c}}}^2_2=1.
$
Define $B_{\vec{c}}=O_{\vec{c}}/\norm{O_{\vec{c}}}_2$. Then, $H[O]$ can be written as
\begin{eqnarray*}
H[U_tOU^\dag_t]
=\sum_{\vec{c}}
\norm{O_{\vec{c}}}^2_2H[U_tB_{\vec{c}}U^\dag_t]
-\sum_{\vec{c}}
\norm{O_{\vec{c}}}^2_2\log\norm{O_{\vec{c}}}^2_2.
\end{eqnarray*}
Hence,
\begin{eqnarray*}
R_M(O,H)=
\sum_{\vec{c}}
\norm{O_{\vec{c}}}^2_2 R_M(B_{\vec{c}},H_S).
\end{eqnarray*}
Then,  by Lemma \ref{lem:magic_ra}, we have 
$|R_M(B_{\vec{c}},H_S)|\leq 4d^k\norm{H_S}_{\infty}$. 
Therefore, we have 
\begin{eqnarray}
|R_M(O,H)|\leq 
\sum_{\vec{c}}
\norm{O_{\vec{c}}}^2_2 |R_M(B_{\vec{c}},H_S)|
\leq 4d^k\norm{H}_{\infty}.
\end{eqnarray}

\end{proof}

In \eqref{eq:exp_dep}, the dependence on the local dimension $d$ occurs as $O(d^k)$. In \S\ref{sec:coh}, the connection between coherence and circuit complexity is studied, where we show that the dependence on the local dimension is $O(k\log d)$. This suggests that a similar bound may also hold
for magic.

\begin{con}\label{con:smal_mag}
Given an $n$-qudit system with the Hamiltonian $H$ acting  on a $k$-qudit subsystem, 
\begin{eqnarray}
|R_M(H, O)|\overset{?}{\leq} 
ck\log(d)\norm{H}_{\infty},
\end{eqnarray}
where $c$ is a constant independent of $k$,$d$, and $n$.
\end{con}

\begin{thm}[\textbf{Magic power bounds the circuit cost}]\label{thm:cost_ma}
The circuit cost of a quantum circuit $U\in \SU(d^n)$ is
lower bounded by the magic power as
follows
\begin{eqnarray}
\Cost(U)
\geq 
\frac{\e}{8d^2\log(\e)}\mathcal{M}[U].
\end{eqnarray}
\end{thm}
\begin{proof}
The proof is almost the same as that of Theorem \ref{thm:cost_In}, which we omit here.

\end{proof}

\begin{cor}
If Conjecture \ref{con:smal_mag} holds, then the circuit cost of a quantum circuit $U\in \SU(d^n)$ is
lower bounded by the magic power as follows
\begin{eqnarray}
\Cost(U)
\geq \frac{c}{\log d}\mathcal{M}[U].
\end{eqnarray}
\end{cor}
\begin{proof}
The proof is almost the same as 
that of Theorem~\ref{thm:cost_ma}, which we omit here.
\end{proof}

\section{Coherence and circuit complexity}\label{sec:coh}

First, let us recall the basic concepts in the resource theory of coherence. 
Given a fixed reference basis $\mathcal{B}=\set{\ket{i}}_i$, any state which is diagonal in the reference basis is called an incoherent state. The set of all incoherent states is denoted as $ \mathcal{I}$. To quantify the coherence in a state, we need to define 
a coherence measure. Examples of such measures include the $l_1$ norm coherence and relative entropy of coherence \cite{baumgratz2014quantifying}. In this work, we focus on the relative entropy of coherence, which is defined as follows 
\begin{eqnarray}
C_r(\rho)=S(\Delta(\rho))-S(\rho),
\end{eqnarray}
where $S(\rho):=-\Tr{\rho\log\rho}$ is the von Neumann entropy of $\rho$ and $\Delta(\cdot):=\sum_i\langle i | \cdot|i\rangle\proj{i}$ is the completely dephasing channel.
This allows us to define the cohering power for a unitary evolution $U$ as:
\begin{eqnarray}
\mathcal{C}_r(U)=
\max_{\rho\in \mathcal{D}((\complex^d)^{\ot n})}
|C_r(U\rho U^\dag)-C_r(\rho)|.
\label{eq:cohering_power}
\end{eqnarray}
where the maximization is taken over all density operators $\rho\in\mathcal{D}((\complex^d)^{\ot n})$.

\begin{defn}[\textit{\textbf{Rate of Coherence}}]
Given an $n$-qudit Hamiltonian $H$ and a quantum state $\rho$, the coherence  rate $R_C(H,\rho)$
is the derivative of the coherence measure with respect to time $t$ at $t=0$:
\begin{eqnarray}
R_C(H,\rho):=
\frac{\d}{\d t}
C_r\left(\e^{-\i tH}\rho \e^{\i tH}\right)\Big|_{t=0}.
\end{eqnarray}
\end{defn}

\begin{lem}\label{lem:form}
Given a Hamiltonian $H$  on an $n$-qudit system and an $n$-qudit quantum state $\rho$, the
coherence rate $R_C(H,\rho)$
can be written  
\begin{eqnarray}
R_C(H,\rho)
=-\i\Tr{[\rho,\log\Delta(\rho)]H}.
\end{eqnarray}
\end{lem}
\begin{proof}
This comes from direct calculation.
\end{proof}

\begin{prop}\label{prop:bound1}
Given an $n$-qudit system with a Hamiltonian $H$ and an $n$-qudit quantum state $\rho\in D((\complex^d)^{\ot n})$,  the coherence rate satisfies the following bound
\begin{eqnarray}
|R_C(H,\rho)|\leq 4\norm{H}_{\infty} 
D_{\max}(\rho\|\Delta(\rho)),
\end{eqnarray}
where 
$D_{\max}$ is the maximal relative entropy defined as 
\begin{eqnarray}
D_{\max}(\rho\|\sigma)=\log\min\set{\lambda:
\rho\leq \lambda \sigma}.
\end{eqnarray}
\end{prop}
\begin{proof}
To prove this result, we need the following lemma.

\begin{lem}\label{lem:cmp}(Mari{\"e}n et al.~\cite{marien2016entanglement})
Given two positive operators $A$ and $B$ with $A\leq B$ and $\Tr{B}=1$, there exists a universal constant $c$ such that 
\begin{eqnarray}
\Tr{|[A,\log B]|}
\leq 4\, h(p),
\end{eqnarray}
where $p=\Tr{A}$, and $h(p)=-p\log p-(1-p)\log(1-p)$. Here, $c$ can be taken to be $4$  \cite{audenaert2014quantum}.
\end{lem}
The proof of Proposition \ref{prop:bound1} is a corollary of the above lemma by taking 
 $A=2^{-D_{\max}(\rho\|\Delta(\rho))}\rho$ and $B=\Delta(\rho)$.
\end{proof}

\begin{thm}
Given an $n$-qudit system with 
the Hamiltonian $H$ acting a $k$-qudit subsystem and an $n$-qudit quantum state $\rho\in D((\complex^d)^{\ot n})$, 
we have 
\begin{eqnarray}
|R_C(H,\rho)|
\leq 4\norm{H}_{\infty}k\log (d).
\end{eqnarray}
\end{thm}

\begin{proof}
Since $H$ acts on a $k$-qudit subsystem,  there exists a subset $S\subset [n]$ with $|S|=k$
such that $H=H_S\ot \mathbb{I}_{S^c}$.
Based on Lemma \ref{lem:form},  we have
\begin{eqnarray*}
R_C(H,\rho)
=-\i\sum_{\vec{z}\in[d]^n}\bra{\vec{z}}[H_S\ot \mathbb{I}_{S^c},\rho]\ket{\vec{z}}
\log p(\vec{z}).
\end{eqnarray*}
Let us decompose $\ket{\vec{z}}=\ket{\vec{x}}\ket{\vec{y}}$, where 
$\vec{x}\in[d]^S$ and $ \vec{y}\in [d]^{S^c}$. Then we have 

\begin{eqnarray*}
&&R_C(H,\rho)\\
&=&-\i\sum_{\vec{x}\in [d]^S, \vec{y}\in [d]^{S^c}}\bra{\vec{x}}\bra{\vec{y}}[H_S\ot \mathbb{I}_{S^c},\rho]\ket{\vec{x}}\ket{\vec{y}}
\log\Tr{\rho\proj{\vec{x}}\ot\proj{\vec{y}}}.
\end{eqnarray*}
Now, let us define a set of $k$-qudit states $\set{\rho_{\vec{y}}}_{\vec{y}}$ as follows
\begin{eqnarray*}
\rho_{\vec{y}}:=\frac{\mathrm{Tr}_{S^c}[\rho\proj{\vec{y}}_{S^c}]}{p_{\vec{y}}},
\end{eqnarray*}
for any $\vec{y}\in [d]^{S^c}$, where the probability $p_{\vec{y}}$ is defined as
\begin{eqnarray*}
p_{\vec{y}}=\Tr{\rho \proj{\vec{y}}_{S^c}\ot \mathbb{I}_{S}}.
\end{eqnarray*}
Note that $\sum_{\vec{y}}p_{\vec{y}}=1$.
Hence, $R_C(H,\rho)$ can be rewritten as
\begin{eqnarray*}
R_C(H,\rho)
&=&-\i\sum_{\vec{x}\in [d]^S}\sum_{\vec{y}\in [d]^{S^c}}
\bra{\vec{x}}[H_S,\rho_{\vec{y}}]\ket{\vec{x}}
p_{\vec{y}}\log (\Tr{\rho_{\vec{y}}\proj{\vec{x}}}p_{\vec{y}})\\
&=&-\i\sum_{\vec{x}\in [d]^S}\sum_{\vec{y}\in [d]^{S^c}}
\bra{\vec{x}}[H_S,\rho_{\vec{y}}]\ket{\vec{x}}
p_{\vec{y}}\log \Tr{\rho_{\vec{y}}\proj{\vec{x}}}\\
&&-\i\sum_{\vec{x}\in [d]^S}\sum_{\vec{y}\in [d]^{S^c}}
\bra{\vec{x}}[H_S,\rho_{\vec{y}}]\ket{\vec{x}}
p_{\vec{y}}\log p_{\vec{y}}.
\end{eqnarray*}
Since 
$
\sum_{\vec{x}\in [d]^S}\
\bra{\vec{x}}[H_S,\rho_{\vec{y}}]\ket{\vec{x}}
=\Tr{[H_S,\rho_{\vec{y}}]}=0,
$
we have
\begin{eqnarray*}
\i\sum_{\vec{x}\in [d]^S}\sum_{\vec{y}\in [d]^{S^c}}
\bra{\vec{x}}[H_S,\proj{\rho_{\vec{y}}}]\ket{\vec{x}}
p_{\vec{y}}\log p_{\vec{y}}=0.
\end{eqnarray*}
Therefore, 
\begin{eqnarray*}
R_C(H,\rho)&=&
-\i\sum_{\vec{x}\in [d]^S}\sum_{\vec{y}\in [d]^{S^c}}
\bra{\vec{x}}[H_S,\rho_{\vec{y}}]\ket{\vec{x}}
p_{\vec{y}}\log\Tr{\rho_{\vec{y}}\proj{\vec{x}}}\\
&=&\sum_{\vec{y}\in [d]^{S^c}}
p_{\vec{y}}\left(
-\i\sum_{\vec{x}\in [d]^S}
\bra{\vec{x}}[H_S,\rho_{\vec{y}}]\ket{\vec{x}}
\log\Tr{\rho_{\vec{y}}\proj{\vec{x}}}
\right)\\
&=&\sum_{\vec{y}\in [d]^{S^c}}
p_{\vec{y}}
R_C(H,\rho_{\vec{y}}).
\end{eqnarray*}
By Proposition \ref{prop:bound1}, we have 
\begin{eqnarray*}
|R_C(H,\rho)|
\leq 4\sum_{\vec{y}\in [d]^{S^c}}
p_{\vec{y}}
\norm{H_S}_{\infty}
D_{\max}(\rho_{\vec{y}}\|\Delta(\rho_{\vec{y}})).
\end{eqnarray*}
Since $\rho_{\vec{y}}$ is a quantum state on a $k$-qudit system, $D_{\max}(\rho_{\vec{y}}\|\Delta(\rho_{\vec{y}}))\leq k\log (d)$. 
Hence, we have
\begin{eqnarray*}
|R_C(H,\rho)|\leq 
4k\norm{H_S}_{\infty}
\log (d).
\end{eqnarray*}
\end{proof}

\begin{thm}\label{thm:cost_co}[Cohering power lower bounds the circuit cost]
The circuit cost of a quantum circuit $U\in \SU(d^n)$ is
lower bounded by the cohering power as follows
\begin{eqnarray}
\Cost(U)
\geq 
\frac{1}{8\log(d)}
\mathcal{C}_r(U).
\end{eqnarray}

\end{thm}
\begin{proof}
The proof is the same as that in Theorem \ref{thm:cost_In}, which we omit here.
\end{proof}

\section{Concluding remarks}
In this work, we investigated the connection between circuit complexity and 
influence, magic, and coherence in quantum circuits. 
Our main result is a lower bound on the circuit complexity by the
circuit sensitivity, magic power, and cohering power of the circuit.

We provided a  characterization of 
scrambling in quantum circuits by the average sensitivity. We gave a characterization of unitaries  with zero circuit sensitivity and showed that such unitaries 
 can be efficiently simulated on a classical computer. In other words, circuits consisting of just these unitaries can yield no quantum advantage. In this regard, our result provides a new understanding of matchgates via sensitivity.
This raises the following interesting question: does the sensitivity of a quantum circuit determine the classical simulation time of the circuit? This is a question we leave for future work. Moreover, it will be 
interesting to develop a framework of quantum algorithmic stability based on sensitivity 
and apply it to quantum differential privacy and generalization capability of quantum machine learning. 

Finally, we also defined a quantum version of the Fourier entropy-influence conjecture, and applied it to establishing a connection between circuit complexity and magic. If the quantum  Fourier entropy-influence conjecture is true, then we can infer that the classical conjecture also holds.

\section*{Acknowledgments}
K.B.\ thanks Xun Gao for useful discussions. This
work was supported in part by the ARO Grant W911NF-19-
1-0302, the ARO MURI Grant W911NF-20-1-0082, including Supplemental Support for Research Trainees (SSRT).

\begin{appendix}

\section{OTOCs}\label{Appen:OTOC}

\begin{lem}[4-point correlator, weight $m$]
If the region $D$ is the last $k$-th qubit, then 
\begin{eqnarray}
\nonumber\mathbb{E}_A\mathbb{E}_{O_A}
\langle O_D(t)O_AO_D(t)O_A\rangle
=
\frac{1}{\binom{n-k}{m}}
\sum^m_{j=0}
\left(-\frac{4}{3}\right)^j
\binom{n-k-j}{m-j}
I^{(j)}_{[n-k]}[O_D(t)],\\
\end{eqnarray}
where $\mathbb{E}_A$ denotes the average over all of the size-$m$ subsets $A\in [n-k]$ so that $O_A$ commutes with $O_D$ at the beginning, 
$\mathbb{E}_{O_A}$ denotes the average over all local Pauli operators with weight $m$ on position $A$, and 
$I^{(j)}_{[n-k]}[O_D(t)]$ is defined as 
\begin{eqnarray}
I^{(j)}_{[n-k]}[O_D(t)]&=&\sum_{S\subset [n-k],|S|=j}I_S[O_D(t)],\\
I_S[O_D(t)]&=&\sum_{S\subset \mathrm{supp}(\vec{a})}P_{O_D(t)}[\vec{a}].
\end{eqnarray}
\end{lem}
\begin{proof}
Let $S$ be a subset of $[n-k]$ with $|S|=m$. The average of all the weight-$m$ Pauli operators with support on $S$ is equal to 
\begin{eqnarray}
\mathbb{E}_{P_S}\langle P_{\vec{a}}P_jP_{\vec{b}}P_j\rangle
=\left(-\frac{1}{3}\right)^{|\supp(\vec{a})\cap S|}\delta_{\vec{a},\vec{b}}.
\end{eqnarray}
Hence,
\begin{eqnarray*}
\allowdisplaybreaks
&&\mathbb{E}_A\mathbb{E}_{O_A}
\langle O_D(t)O_AO_D(t)O_A\rangle\\
&=&\frac{1}{\binom{n-k}{m}}
\sum_{S\subset [n-k]}\sum_{\vec{a}}\left(-\frac{1}{3}\right)^{|\supp(\vec{a})\cap S|}
P_{O_D(t)}[\vec{a}]\\
&=&\frac{1}{\binom{n-k}{m}}
\sum_{S\subset [n-k]}\sum_{\vec{a}}\left(-\frac{1}{3}\right)^{|\supp(\vec{a})\cap S|}
P_{O_D(t)}[\vec{a}]\\
&=&\frac{1}{3^m}\frac{1}{\binom{n-k}{m}}
\sum_{\vec{a}}\sum_{S\subset [n-k]}3^m\left(-\frac{1}{3}\right)^{|\supp(\vec{a})\cap S|}
P_{O_D(t)}[\vec{a}]\\
&=&\frac{1}{3^m}\frac{1}{\binom{n-k}{m}}
\sum_{\vec{a}}
3^m\sum^m_{j=0}\left(-\frac{1}{3}\right)^j\binom{|\supp(\vec{a})\cap [n-k]|}{j}\\
&&\qquad\qquad\qquad\times
\binom{n-k-|\supp(\vec{a})\cap [n-k]|}{m-j}
P_{O_D(t)}[\vec{a}]
\end{eqnarray*}
Let us introduce the Krawtchouk polynomial $K_m(x;n,q)$, which is defined as follows:
\begin{eqnarray}
K_m(x;n,q)=\sum^m_{j=0}
(-1)^j(q-1)^{m-j}
\binom{x}{j}
\binom{n-x}{m-j}.
\end{eqnarray}
This can be rewritten as 
\begin{eqnarray}
K_m(x;n,q)=\sum^m_{j=0}
(-q)^j(q-1)^{m-j}
\binom{n-j}{m-j}
\binom{x}{j}.
\end{eqnarray}
Then the above equation  equals

\begin{eqnarray*}
&&\mathbb{E}_A\mathbb{E}_{O_A}
\langle O_D(t)O_AO_D(t)O_A\rangle\\
&=&\frac{1}{3^m}\frac{1}{\binom{n-k}{m}}
\sum_{\vec{a}}
K_m(|\supp(\vec{a})\cap [n-k]|;n-k,4)
P_{O_D(t)}[\vec{a}]\\
&=&\frac{1}{3^m}\frac{1}{\binom{n-k}{m}}
\sum_{\vec{a}}
\sum^m_{j=0}(-4)^j
3^{m-j}
\binom{n-k-j}{m-j}\binom{|\supp(\vec{a})\cap [n-k]|}{j}
P_{O_D(t)}[\vec{a}]\\
&=&\frac{1}{\binom{n-k}{m}}
\sum^m_{j=0}\left(-\frac{4}{3}\right)^j
\binom{n-k-j}{m-j}\\
&&\qquad\qquad\qquad\times
\sum_{\vec{a}:|\supp(\vec{a})\cap [n-k]|)\geq j}\binom{|\supp(\vec{a})\cap [n-k]|}{j}
P_{O_D(t)}[\vec{a}]\\
&=&\frac{1}{\binom{n-k}{m}}
\sum^m_{j=0}\left(-\frac{4}{3}\right)^j
\binom{n-k-j}{m-j}\\
&&\qquad\qquad\qquad\times
\sum_{S:S\subset [n-k], |S|=j}
\sum_{\vec{a}:S\subset \supp(\vec{a})}
P_{O_D(t)}[\vec{a}]\\
&=&\frac{1}{\binom{n-k}{m}}
\sum^m_{j=0}\left(-\frac{4}{3}\right)^j
\binom{n-k-j}{m-j}I^j_{[n-k]}[O_D(t)].
\end{eqnarray*}

\end{proof}

\begin{lem}[8-point correlator]
If the region $D$ is the last $k$-th qubit, then 
\begin{eqnarray}
\nonumber&&\mathbb{E}_A\mathbb{E}_{O_A}
\langle O_D(t)O_AO_D(t)O_A O_D(t)O_AO_D(t)O_A\rangle\\
&=&
\norm{O_D(t)*O_D(t)}^2_2\left[1
-\frac{4}{3}\frac{1}{n-k}\sum_{j\in [n-k]}
I_j[O_D(t)*O_D(t)]\right],
\end{eqnarray}
where $\mathbb{E}_A$ denotes the average over all of the positions $j\in [n-k]$ so that $O_A$ commutes with $O_D$ at the beginning, and 
$\mathbb{E}_{O_A}$ denotes the average over all local non-identity Pauli operators on position $D$. The convolution 
$O_D(t)*O_D(t)$ is defined in \eqref{eq:conv}.
\end{lem}

\begin{proof}
Let us express the operator $O=\sum_{\vec{a}}\hat{f}(\vec{a})P_{\vec{a}}$. Then, the correspond generalized Wigner function 
$f$ is defined as follows
\begin{eqnarray}
f(\vec{x})
=\sum_{\vec{a}}\hat{f}(\vec{a})
(-1)^{\inner{\vec{x}}{\vec{a}}_s},
\end{eqnarray}
where the inner product  $\inner{\cdot}{\cdot}_s$ denotes the symplectic inner product. (See Appendix \ref{appen:wign} for a brief 
introduction of the generalized Wigner function and symplectic Fourier transformation.)

Let us first consider the average of $\langle P_{\vec{a}}P_jP_{\vec{b}}P_jP_{\vec{c}}P_jP_{\vec{d}}P_j\rangle$. 
It is easy to verify that 
\begin{eqnarray*}
&&\mathbb{E}_{P_j}\langle P_{\vec{a}}P_jP_{\vec{b}}P_jP_{\vec{c}}P_jP_{\vec{d}}P_j\rangle\\
&=&\frac{1}{3}\left[
\langle P_{\vec{a}}X_jP_{\vec{b}}X_jP_{\vec{c}}X_jP_{\vec{d}}X_j\rangle+
\langle P_{\vec{a}}Y_jP_{\vec{b}}Y_jP_{\vec{c}}Y_jP_{\vec{d}}Y_j\rangle+
\langle P_{\vec{a}}Z_jP_{\vec{b}}Z_jP_{\vec{c}}Z_jP_{\vec{d}}Z_j\rangle
\right]\\
&=&\frac{1}{3}(4\delta_{b_j+d_j,0}-1)\delta_{\vec{a}+\vec{b}+\vec{c}+\vec{d},\vec{0}}\\
&=&\left[1-\frac{4}{3}|b_j+d_j|\right]\delta_{\vec{a}+\vec{b}+\vec{c}+\vec{d},\vec{0}}.
\end{eqnarray*}
Therefore, we have 
\begin{eqnarray*}
&&\mathbb{E}_A\mathbb{E}_{O_A}
\langle O_D(t)O_AO_D(t)O_A O_D(t)O_AO_D(t)O_A\rangle\\
&=&\frac{1}{n-k}\sum_{j\in [n-k]}
\sum_{\vec{a},\vec{b},\vec{c},\vec{d}}
\left[1-\frac{4}{3}|b_j+d_j|\right]\delta_{\vec{a}+\vec{b}+\vec{c}+\vec{d},\vec{0}}
\hat{f}_{O_D}(\vec{a})\hat{f}_{O_D}(\vec{b})\hat{f}_{O_D}(\vec{c})\hat{f}_{O_D}(\vec{d})\\
&=&\frac{1}{n-k}\sum_{j\in [n-k]}
\sum_{\vec{a},\vec{b},\vec{c},\vec{d}}\delta_{\vec{a}+\vec{b}+\vec{c}+\vec{d},\vec{0}}
\hat{f}_{O_D}(\vec{a})\hat{f}_{O_D}(\vec{b})\hat{f}_{O_D}(\vec{c})\hat{f}_{O_D}(\vec{d})\\
&&-\frac{4}{3}\frac{1}{n-k}\sum_{j\in [n-k]}
\sum_{\vec{a},\vec{b},\vec{c},\vec{d}}|b_j+d_j|\delta_{\vec{a}+\vec{b}+\vec{c}+\vec{d},\vec{0}}
\hat{f}_{O_D}(\vec{a})\hat{f}_{O_D}(\vec{b})\hat{f}_{O_D}(\vec{c})\hat{f}_{O_D}(\vec{d})\\
&=&\sum_{\vec{a},\vec{b},\vec{c},\vec{d}}\delta_{\vec{a}+\vec{b}+\vec{c}+\vec{d},\vec{0}}
\hat{f}_{O_D}(\vec{a})\hat{f}_{O_D}(\vec{b})\hat{f}_{O_D}(\vec{c})\hat{f}_{O_D}(\vec{d})\\
&&
-\frac{4}{3}\frac{1}{n-k}\sum_{j\in [n-k]}
\sum_{\vec{a},\vec{b},\vec{c},\vec{d}}|b_j+d_j|\delta_{\vec{a}+\vec{b}+\vec{c}+\vec{d},\vec{0}}
\hat{f}_{O_D}(\vec{a})\hat{f}_{O_D}(\vec{b})\hat{f}_{O_D}(\vec{c})\hat{f}_{O_D}(\vec{d}).\\
\end{eqnarray*}
Let us compute the two terms separately. 
First, 
\begin{eqnarray*}
&&\sum_{\vec{a},\vec{b},\vec{c},\vec{d}}\delta_{\vec{a}+\vec{b}+\vec{c}+\vec{d},\vec{0}}
\hat{f}_{O_D}(\vec{a})\hat{f}_{O_D}(\vec{b})\hat{f}_{O_D}(\vec{c})\hat{f}_{O_D}(\vec{d})\\
&=&\sum_{\vec{a},\vec{b},\vec{c}}
\hat{f}(\vec{a})\hat{f}_{O_D}(\vec{a}+\vec{b})\hat{f}_{O_D}(\vec{a}+\vec{c})\hat{f}_{O_D}(\vec{a}+\vec{b}+\vec{c})\\
&=&\sum_{\vec{b}}\left[\sum_{\vec{a}}\hat{f}_{O_D}(\vec{a})\hat{f}_{O_D}(\vec{a}+\vec{b})\right]
\left[\sum_{\vec{c}}\hat{f}_{O_D}(\vec{c})\hat{f}_{O_D}(\vec{c}+\vec{b})\right]\\
&=&\sum_{\vec{b}}
\left(\mathbb{E}_{\vec{a}}|f_{O_D}(\vec{a})|^2(-1)^{\inner{\vec{a}}{\vec{b}}_s}\right)^2
\\
&=&\mathbb{E}_{\vec{a}}|f_{O_D}(\vec{a})|^4\\
&=&\norm{O_D*O_D}^2_2,
\end{eqnarray*}
where the convolution 
$O*P$ satisfies
\begin{eqnarray}
f_{O*P}
=f_Of_P.
\end{eqnarray}

Besides, 
\begin{eqnarray*}
&&\frac{1}{n-k}\sum_{j\in [n-k]}
\sum_{\vec{a},\vec{b},\vec{c},\vec{d}}|b_j+d_j|\delta_{\vec{a}+\vec{b}+\vec{c}+\vec{d},\vec{0}}
\hat{f}_{O_D}(\vec{a})\hat{f}_{O_D}(\vec{b})\hat{f}_{O_D}(\vec{c})\hat{f}_{O_D}(\vec{d})\\
&=&\frac{1}{n-k}
\sum_{\vec{a},\vec{b},\vec{c},\vec{d}}\left[\left|[n-k]\cap \left(\vec{b}+\vec{d}\right)\right|\right]\delta_{\vec{a}+\vec{b}+\vec{c}+\vec{d},\vec{0}}
\hat{f}_{O_D}(\vec{a})\hat{f}_{O_D}(\vec{b})\hat{f}_{O_D}(\vec{c})\hat{f}_{O_D}(\vec{d})\\
&=&
\frac{1}{n-k}\sum_{\vec{a},\vec{b},\vec{c}}
|[n-k]\cap \vec{c}|
\hat{f}_{O_D}(\vec{a})\hat{f}_{O_D}(\vec{a}+\vec{b})\hat{f}_{O_D}(\vec{a}+\vec{c})\hat{f}_{O_D}(\vec{a}+\vec{b}+\vec{c})\\
&=&\frac{1}{n-k}\sum_{\vec{c}}|[n-k]\cap \vec{c}|
\left[\sum_{\vec{a}}\hat{f}_{O_D}(\vec{a})\hat{f}_{O_D}(\vec{a}+\vec{c})\right]
\left[\sum_{\vec{b}}\hat{f}_{O_D}(\vec{b})\hat{f}_{O_D}(\vec{b}+\vec{c})\right]\\
&=&
\frac{1}{n-k}\sum_{\vec{c}}|[n-k]\cap \vec{c}|
\left(\mathbb{E}_{\vec{a}}|f_{O_D}(\vec{a})|^2(-1)^{\inner{\vec{a}}{\vec{c}}_s}\right)^2,
\end{eqnarray*}
where it is easy to verify that
\begin{eqnarray}
I_j[O*O]
=\sum_{\vec{c}:c_j\neq 0}
\left|\mathbb{E}_{\vec{a}} |f_{O_D}(\vec{a})|^2(-1)^{\inner{\vec{a}}{\vec{c}}_s}\right|^2.
\end{eqnarray}

\end{proof}

\begin{lem}[8-point correlator, weight $m$]
If the region $D$ is the last $k$-th qubit, then 
\begin{eqnarray}
\nonumber&&\mathbb{E}_A\mathbb{E}_{O_A}
\langle O_D(t)O_AO_D(t)O_A\langle O_D(t)O_AO_D(t)O_A\rangle\\
\nonumber&=&\norm{O_D(t)*O_D(t)}^2_2
\frac{1}{\binom{n-k}{m}}
\sum^m_{j=0}
\left(-\frac{4}{3}\right)^j
\binom{n-k-j}{m-j}
I^{(j)}_{[n-k]}[O_D(t)*O_D(t)],\\
\end{eqnarray}
where $\mathbb{E}_A$ denotes the average over all of the size-$m$ subsets $A\in [n-k]$ so that $O_A$ commutes with $O_D$ at the beginning, 
$\mathbb{E}_{O_A}$ denotes the average over all local Pauli operators with weight $m$ on position $A$, and 
$I^{(j)}_{[n-k]}[O_D(t)*O_D(t)]$ is defined as above. 
\end{lem}
\begin{proof}
Since $O=\sum_{\vec{a}}\hat{f}(\vec{a})P_{\vec{a}}$, 
let us first consider the average of $\langle P_{\vec{a}}P_SP_{\vec{b}}P_SP_{\vec{c}}P_SP_{\vec{d}}P_S\rangle$. 
It is easy to verify that 
\begin{eqnarray}
\mathbb{E}_{P_S}\langle P_{\vec{a}}P_SP_{\vec{b}}P_SP_{\vec{c}}P_SP_{\vec{d}}P_S\rangle
=\left(
-\frac{1}{3}
\right)^{|\supp(\vec{b}+\vec{d})\cap S|}
\delta_{\vec{a}+\vec{b}+\vec{c}+\vec{d},\vec{0}}.
\end{eqnarray}
Hence, we have 
\begin{eqnarray*}
&&\mathbb{E}_A\mathbb{E}_{O_A}
\langle O_D(t)O_AO_D(t)O_AO_D(t)O_AO_D(t)O_A\rangle\\
&=&\frac{1}{\binom{n-k}{m}}
\sum_{S\subset [n-k]}\sum_{\vec{a},\vec{b},\vec{c},\vec{d}}\left(-\frac{1}{3}\right)^{|\supp(\vec{b}+\vec{d})\cap S|}\delta_{\vec{a}+\vec{b}+\vec{c}+\vec{d}, 0}
\hat{f}_{O_D}(\vec{a})\hat{f}_{O_D}(\vec{b})\hat{f}_{O_D}(\vec{c})\hat{f}_{O_D}(\vec{d})\\
&=&\frac{1}{3^m}\frac{1}{\binom{n-k}{m}}
\sum_{\vec{a},\vec{b},\vec{c},\vec{d}}\sum_{S\subset [n-k]}3^m\left(-\frac{1}{3}\right)^{|\supp(\vec{b}+\vec{d})\cap S|}
\delta_{\vec{a}+\vec{b}+\vec{c}+\vec{d}, 0}\\
&&\qquad\qquad\qquad\qquad\times\hat{f}_{O_D}(\vec{a})\hat{f}_{O_D}(\vec{b})\hat{f}_{O_D}(\vec{c})\hat{f}_{O_D}(\vec{d})\\
&=&\frac{1}{3^m}\frac{1}{\binom{n-k}{m}}
\sum_{\vec{a},\vec{b},\vec{c},\vec{d}}
3^m\sum^m_{j=0}\left(-\frac{1}{3}\right)^j\binom{|[n-k]\cap \supp(\vec{b}+\vec{d})|}{j}\\
&&\qquad\qquad\qquad\qquad\times
\binom{n-k-|[n-k]\cap \supp(\vec{b}+\vec{d})|}{m-j}\\
&&\qquad\qquad\qquad\qquad\times
\delta_{\vec{a}+\vec{b}+\vec{c}+\vec{d}, 0}\hat{f}_{O_D}(\vec{a})\hat{f}_{O_D}(\vec{b})\hat{f}_{O_D}(\vec{c})\hat{f}_{O_D}(\vec{d})\\
&=&\frac{1}{3^m}\frac{1}{\binom{n-k}{m}}
\sum_{\vec{a},\vec{b},\vec{c}, \vec{d}}
K_m(|[n-k]\cap \supp(\vec{b}+\vec{d})|;n-k,4)\\
&&\qquad\qquad\qquad\times
\delta_{\vec{a}+\vec{b}+\vec{c}+\vec{d}, 0}\hat{f}_{O_D}(\vec{a})\hat{f}_{O_D}(\vec{b})\hat{f}_{O_D}(\vec{c})\hat{f}_{O_D}(\vec{d})\\
&=&\frac{1}{3^m}\frac{1}{\binom{n-k}{m}}
\sum_{\vec{a},\vec{b},\vec{c}}
K_m(|[n-k]\cap \supp(\vec{c})|;n-k,4)\\
&&\qquad\qquad\qquad\times
\hat{f}_{O_D}(\vec{a})\hat{f}_{O_D}(\vec{a}+\vec{b})\hat{f}_{O_D}(\vec{a}+\vec{c})\hat{f}_{O_D}(\vec{a}+\vec{b}+\vec{c})
\\
&=&
\frac{1}{3^m}\frac{1}{\binom{n-k}{m}}
\sum_{\vec{c}}K_m(|[n-k]\cap \supp(\vec{c})|;n-k,4)\\
&&\qquad\qquad\times
\sum_{\vec{a},\vec{b}}
\hat{f}_{O_D}(\vec{a})\hat{f}_{O_D}(\vec{a}+\vec{b})\hat{f}_{O_D}(\vec{a}+\vec{c})\hat{f}_{O_D}(\vec{a}+\vec{b}+\vec{c}).
\end{eqnarray*}
Since 
\begin{eqnarray}
\hat{f}(\vec{a})\hat{f}_{O_D}(\vec{a}+\vec{b})\hat{f}_{O_D}(\vec{a}+\vec{c})\hat{f}_{O_D}(\vec{a}+\vec{b}+\vec{c})
=\left(\mathbb{E}_{\vec{a}}|f(\vec{a})|^2(-1)^{\inner{\vec{a}}{\vec{c}}_s}\right)^2,
\end{eqnarray}
we have 
\begin{eqnarray*}
&&\mathbb{E}_A\mathbb{E}_{O_A}
\langle O_D(t)O_AO_D(t)O_AO_D(t)O_AO_D(t)O_A\rangle\\
&=&
\frac{1}{3^m}\frac{1}{\binom{n-k}{m}}
\sum_{\vec{c}}K_m(|[n-k]\cap \supp(\vec{c})|;n-k,4)\\
&&\qquad\qquad\times
\sum_{\vec{a},\vec{b}}
\hat{f}_{O_D}(\vec{a})\hat{f}_{O_D}(\vec{a}+\vec{b})\hat{f}_{O_D}(\vec{a}+\vec{c})\hat{f}_{O_D}(\vec{a}+\vec{b}+\vec{c})
\\
&=&\frac{1}{3^m}\frac{1}{\binom{n-k}{m}}
\sum_{\vec{c}}K_m(|[n-k]\cap \supp(\vec{c})|;n-k,4)
\left(\mathbb{E}_{\vec{a}}|f(\vec{a})|^2(-1)^{\inner{\vec{a}}{\vec{c}}_s}\right)^2
\\
&=&\frac{1}{3^m}\frac{1}{\binom{n-k}{m}}
\sum_{\vec{c}}
\sum^m_{j=0}(-4)^j
3^{m-j}
\binom{n-k-j}{m-j}\binom{|[n-k]\cap \supp(\vec{c})|}{j}
|\hat{f}_{O*O}(\vec{c})|^2\\
&=&\norm{O_D(t)*O_D(t)}^2_2\frac{1}{\binom{n-k}{m}}
\sum^m_{j=0}\left(-\frac{4}{3}\right)^j
\binom{n-k-j}{m-j}I^j_{[n-k]}[O_D(t)*O_D(t)].
\end{eqnarray*}

\end{proof}

\section{Boolean Fourier entropy-influence conjecture}\label{Appen:FourierAnalysisClassical}
In this section, we give a brief introduction to the Fourier entropy-influence conjecture for Boolean functions. Boolean functions, defined as functions $f:\set{-1,1}^n\to \set{-1,1}$ (or $ \real$), are a basic object in theoretical computer science. The inner product between Boolean functions is defined as
 \begin{eqnarray*}
 \inner{f}{g}: =\mathbb{E}_xf(x)g(x),
 \end{eqnarray*}
 where $\mathbb{E}_{x}:=\frac{1}{2^n}\sum_{x\in \set{\pm}^n}$.
Each Boolean function $f$ has the following Fourier expansion
\begin{eqnarray*}
f(x)=\sum_{S\subset [n]}\hat{f}(S)x_S,
\end{eqnarray*}
where the parity functions $x_S:=\prod_{i\in S}x_i$, and the Fourier coefficients $\hat{f}(S)=\inner{f}{x_S}=\mathbb{E}_{x\in\set{-1,1}^n}f(x)x_S$. Parseval's identity tells us that 
\begin{eqnarray*}
\mathbb{E}_{x\in \set{\pm}^n}f(x)^2
=\sum_S\hat{f}(S)^2.
\end{eqnarray*}
Let us define the discrete derivative $D_j[f]$  as 
$D_j[f](x)=(f(x)-f(x\oplus e_j))/2$, where $x\oplus e_j$ denotes the flip from $x_j$ to 
$-x_j$. Then
the $j$-th local influence $I_j$ is defined as the $l_2$ norm of the discrete derivative $D_i[f]$:
\begin{eqnarray*}
I_j[f]=
\mathbb{E}_{x\in \set{\pm 1}^n}|D_j[f](x)|^2,
\end{eqnarray*}
which can also be written as 
\begin{eqnarray*}
I_j[f]=\sum_{S:j\in S}
\hat{f}(S)^2|S|,
\end{eqnarray*}
where $|S| $ denotes the size of the subset $S$. 
The total influence of the Boolean function is defined as 
$I[f]=\sum_{j\in[n]} I_j[f]$, which can also be written as 
\begin{eqnarray*}
I[f]=\sum_{S\subset [n]}\hat{f}(S)^2|S|.
\end{eqnarray*}
Assume that $\norm{f}_2=1$. Then, $\sum_S\hat{f}(S)^2=1$ and the Fourier entropy of the Boolean function $f$ is defined as
\begin{eqnarray*}
H[f]=\sum_{S\subset [n]}|\hat{f}(S)|^2\log\frac{1}{|\hat{f}(S)|^2},
\end{eqnarray*}
and the min Fourier entropy $H_{\infty}$ is defined as 
\begin{eqnarray*}
H_{\infty}[f]=\min_{S\subset [n]}\log\frac{1}{|\hat{f}(S)|^2}.
\end{eqnarray*}
One of most important open problems in the analysis of Boolean functions is proving the Fourier entropy-influence (FEI) conjecture that was proposed by Friedgut and Kalai~\cite{friedgut1996every}.

\begin{con}[FEI conjecture]
There exists a universal constant $c$ such that, for all $f:\set{-1,1}^n\to \set{-1,1}$,
\begin{eqnarray}
H[f]\leq cI[f].
\end{eqnarray}
\end{con}

A natural extension of the FEI conjecture is the following Fourier min-entropy-influence conjecture, which follows from the fact that
$H_{\min}[f]\leq H[f]$.

\begin{con}[FMEI conjecture]
There exists a universal constant $c$ such that, for all $f:\set{-1,1}^n\to \set{-1,1}$,
\begin{eqnarray}
H_{\min}[f]\leq cI[f].
\end{eqnarray}
\end{con}

Although both the FEI and FMEI conjectures remain  open, several significant steps have been made to prove these conjectures; see \cite{kelman2020towards,chakraborty2016upper,klivans2010mansour,donnell2013composition,donnell2011fourier,shalev2018fourier,wan2014decision,gopalan2016degree}.

\section{Discrete Wigner function and symplectic Fourier transformation}\label{appen:wign}
We introduce some basics on the Fourier analysis of the discrete Wigner function. 
The discrete Wigner function was proposed for the odd-dimensional case, 
and one well-known result for odd-dimensional discrete Wigner functions is the discrete Hudson theorem, which states that 
any given pure state is a stabilizer state if and only if its Wigner function is nonnegative\cite{gross2006hudson}. Here, we  generalize the definition of the discrete Wigner function to the qubit case, where the discrete Hudson theorem may not 
hold.  

Let us define the generalized phase point operator as 
follows
\begin{eqnarray}
A_{\vec{a}}
=\sum_{\vec{b}}
P_{\vec{b}}(-1)^{\inner{\vec{a}}{\vec{b}}_s},
\end{eqnarray}
where $P_{\vec{b}}$ is an $n$-qubit Pauli operator
and $\inner{\vec{a}}{\vec{b}}_s$ denotes the symplectic inner product. 
Hence, given an observable $O$ (or a quantum state), the (generalized) discrete Wigner function 
$f$ is defined as follows
\begin{eqnarray}
f_O(\vec{a})
=\inner{O}{A_{\vec{a}}},
\end{eqnarray}
which can also be written as follows
\begin{eqnarray}
f_O(\vec{a})
=\sum_{\vec{b}}\inner{P_{\vec{b}}}{O}(-1)^{\inner{\vec{a}}{\vec{b}}_s}
=\sum_{\vec{b}}O_{\vec{b}}(-1)^{\inner{\vec{a}}{\vec{b}}_s}.
\end{eqnarray}
Hence, the Pauli coefficient $O_{\vec{b}}$ is the symplectic Fourier transform of 
the discrete Wigner function, i.e., 
\begin{eqnarray}
O_{\vec{b}}
=\hat{f}_O(\vec{b})
=\mathbb{E}_{\vec{a}}
f_O(\vec{a})(-1)^{\inner{\vec{a}}{\vec{b}}_s}.
\end{eqnarray}
Parseval's identity tells us that 
\begin{eqnarray}
\mathbb{E}_{\vec{a}}
f(\vec{a})^2
=\sum_{\vec{b}}
\hat{f}_O(\vec{b})^2
=\sum_{\vec{b}}|O_{\vec{b}}|^2.
\end{eqnarray}

To consider the higher-order OTOC, we need to use the convolution of two observables. 
We define the convolution of 
two observables $O_1$ and $O_2$ as follows
\begin{eqnarray}\label{eq:conv}
f_{O_1*O_2}
=f_{O_1}f_{O_2}.
\end{eqnarray}
Hence
\begin{eqnarray}
\hat{f}_{O_1*O_2}(\vec{b})
=
\mathbb{E}_{\vec{a}}
f_{O_1}(\vec{a})
f_{O_2}(\vec{a})
(-1)^{\inner{\vec{a}}{\vec{b}}_s}.
\end{eqnarray}

\end{appendix}

\bibliographystyle{ieeetr}
\bibliography{cost-lit}{}

\begin{thebibliography}{10}

\bibitem{nielsen2010quantum}
M.~A. Nielsen and I.~L. Chuang, {\em Quantum Computation and Quantum
  Information}.
\newblock Cambridge University Press, 2010.

\bibitem{kitaev2002classical}
A.~Kitaev, A.~Shen, and M.~Vyalyi, {\em Classical and quantum computation}.
\newblock American Mathematical Society, 2002.

\bibitem{aaronson2016complexity}
S.~Aaronson, ``The complexity of quantum states and transformations: from
  quantum money to black holes,'' {\em arXiv preprint arXiv:1607.05256}, 2016.

\bibitem{nielsen2006geometric}
M.~A. Nielsen, ``A geometric approach to quantum circuit lower bounds,'' {\em
  Quantum Information \& Computation}, vol.~6, no.~3, pp.~213--262, 2006.

\bibitem{nielsen2006optimal}
M.~A. Nielsen, M.~R. Dowling, M.~Gu, and A.~C. Doherty, ``Optimal control,
  geometry, and quantum computing,'' {\em Phys. Rev. A}, vol.~73, p.~062323,
  Jun 2006.

\bibitem{nielsen2006quantum}
M.~A. Nielsen, M.~R. Dowling, M.~Gu, and A.~C. Doherty, ``Quantum computation
  as geometry,'' {\em Science}, vol.~311, no.~5764, pp.~1133--1135, 2006.

\bibitem{dowling2008geometry}
M.~R. Dowling and M.~A. Nielsen, ``The geometry of quantum computation,'' {\em
  Quantum Information \& Computation}, vol.~8, no.~10, pp.~861--899, 2008.

\bibitem{brown2017quantum}
A.~R. Brown, L.~Susskind, and Y.~Zhao, ``Quantum complexity and negative
  curvature,'' {\em Phys. Rev. D}, vol.~95, p.~045010, Feb 2017.

\bibitem{susskind2016typical}
L.~Susskind, ``The typical-state paradox: diagnosing horizons with
  complexity,'' {\em Fortschritte der Physik}, vol.~64, no.~1, pp.~84--91,
  2016.

\bibitem{brown2016holographic}
A.~R. Brown, D.~A. Roberts, L.~Susskind, B.~Swingle, and Y.~Zhao, ``Holographic
  complexity equals bulk action?,'' {\em Phys. Rev. Lett.}, vol.~116,
  p.~191301, May 2016.

\bibitem{chapman2018toward}
S.~Chapman, M.~P. Heller, H.~Marrochio, and F.~Pastawski, ``Toward a definition
  of complexity for quantum field theory states,'' {\em Phys. Rev. Lett.},
  vol.~120, p.~121602, Mar 2018.

\bibitem{brandao2021models}
F.~G. Brand\~ao, W.~Chemissany, N.~Hunter-Jones, R.~Kueng, and J.~Preskill,
  ``Models of quantum complexity growth,'' {\em PRX Quantum}, vol.~2,
  p.~030316, Jul 2021.

\bibitem{jefferson2017circuit}
R.~A. Jefferson and R.~C. Myers, ``Circuit complexity in quantum field
  theory,'' {\em Journal of High Energy Physics}, vol.~2017, no.~10, pp.~1--81,
  2017.

\bibitem{takayanagi2018holographic}
T.~Takayanagi, ``Holographic spacetimes as quantum circuits of
  path-integrations,'' {\em Journal of High Energy Physics}, vol.~2018, no.~12,
  pp.~1--37, 2018.

\bibitem{bhattacharyyaj2018circuit}
A.~Bhattacharyya, A.~Shekar, and A.~Sinha, ``Circuit complexity in interacting
  qfts and rg flows,'' {\em Journal of High Energy Physic}, vol.~2018, p.~140,
  Oct 2018.

\bibitem{chagnet2022complexity}
N.~Chagnet, S.~Chapman, J.~de~Boer, and C.~Zukowski, ``Complexity for conformal
  field theories in general dimensions,'' {\em Phys. Rev. Lett.}, vol.~128,
  p.~051601, Jan 2022.

\bibitem{bhattacharyya2022complexity}
A.~Bhattacharyya, G.~Katoch, and S.~R. Roy, ``Complexity of warped conformal
  field theory,'' {\em arXiv preprint arXiv:2202.09350}, 2022.

\bibitem{couch2021circuit}
J.~Couch, Y.~Fan, and S.~Shashi, ``Circuit complexity in topological quantum
  field theory,'' {\em arXiv preprint arXiv:2108.13427}, 2021.

\bibitem{halpern2021resource}
N.~Y. Halpern, N.~B. Kothakonda, J.~Haferkamp, A.~Munson, J.~Eisert, and
  P.~Faist, ``Resource theory of quantum uncomplexity,'' {\em arXiv preprint
  arXiv:2110.11371}, 2021.

\bibitem{eisert2021entangling}
J.~Eisert, ``Entangling power and quantum circuit complexity,'' {\em Phys. Rev.
  Lett.}, vol.~127, p.~020501, Jul 2021.

\bibitem{odonnell2014analysis}
R.~O'Donnell, {\em Analysis of {B}oolean functions}.
\newblock Cambridge University Press, 2014.

\bibitem{kahn1988influence}
J.~Kahn, G.~Kalai, and N.~Linial, ``The influence of variables on {B}oolean
  functions,'' in {\em [Proceedings 1988] 29th Annual Symposium on Foundations
  of Computer Science}, pp.~68--80, 1988.

\bibitem{linial1989constant}
N.~Linial, Y.~Mansour, and N.~Nisan, ``Constant depth circuits, {F}ourier
  transform, and learnability,'' in {\em 30th Annual Symposium on Foundations
  of Computer Science}, pp.~574--579, 1989.

\bibitem{boppana1997average}
R.~B. Boppana, ``The average sensitivity of bounded-depth circuits,'' {\em
  Information Processing Letters}, vol.~63, no.~5, pp.~257--261, 1997.

\bibitem{jukna2012boolean}
S.~Jukna, {\em Boolean Function Complexity: Advances and Frontiers}.
\newblock Berlin, Germany: Springer, 2012.

\bibitem{lovett2011bounded}
S.~Lovett and E.~Viola, ``Bounded-depth circuits cannot sample good codes,'' in
  {\em 2011 IEEE 26th Annual Conference on Computational Complexity},
  pp.~243--251, 2011.

\bibitem{shi2000lower}
Y.~Shi, ``Lower bounds of quantum black-box complexity and degree of
  approximating polynomials by influence of {B}oolean variables,'' {\em
  Information Processing Letters}, vol.~75, no.~1, pp.~79--83, 2000.

\bibitem{montanaro2010quantum}
A.~Montanaro and T.~J. Osborne, ``Quantum boolean functions,'' {\em Chicago
  Journal of Theoretical Computer Science}, vol.~2010, January 2010.

\bibitem{carlen1993optimal}
E.~A. Carlen and E.~H. Lieb, ``{Optimal hypercontractivity for Fermi fields and
  related noncommutative integration inequalities},'' {\em Communications in
  Mathematical Physics}, vol.~155, no.~1, pp.~27 -- 46, 1993.

\bibitem{valiant2002quantum}
L.~G. Valiant, ``Quantum circuits that can be simulated classically in
  polynomial time,'' {\em SIAM Journal on Computing}, vol.~31, no.~4,
  pp.~1229--1254, 2002.

\bibitem{bravyi2005lagrangian}
S.~Bravyi, ``Lagrangian representation for fermionic linear optics,'' {\em
  Quantum Information \& Computation}, vol.~5, no.~3, pp.~216--238, 2005.

\bibitem{divincenzo2004fermionic}
D.~P. DiVincenzo and B.~M. Terhal, ``Fermionic linear optics revisited,'' {\em
  Foundations of Physics}, vol.~35, pp.~1967--1984, 2004.

\bibitem{terhal2002classical}
B.~M. Terhal and D.~P. DiVincenzo, ``Classical simulation of
  noninteracting-fermion quantum circuits,'' {\em Phys. Rev. A}, vol.~65,
  p.~032325, Mar 2002.

\bibitem{jozsa2008matchgates}
R.~Jozsa and A.~Miyake, ``Matchgates and classical simulation of quantum
  circuits,'' {\em Proc. R. Soc. Lond. A}, vol.~464, p.~3089–3106, Jul 2008.

\bibitem{brod2016efficient}
D.~J. Brod, ``Efficient classical simulation of matchgate circuits with
  generalized inputs and measurements,'' {\em Phys. Rev. A}, vol.~93,
  p.~062332, Jun 2016.

\bibitem{hebenstreit2019all}
M.~Hebenstreit, R.~Jozsa, B.~Kraus, S.~Strelchuk, and M.~Yoganathan, ``All pure
  fermionic non-{G}aussian states are magic states for matchgate
  computations,'' {\em Phys. Rev. Lett.}, vol.~123, p.~080503, Aug 2019.

\bibitem{gottesman1998heisenberg}
D.~Gottesman, ``The {H}eisenberg representation of quantum computers,'' in {\em
  Proc. XXII International Colloquium on Group Theoretical Methods in Physics,
  1998}, pp.~32--43, 1998.

\bibitem{nest2010classical}
M.~V. den Nest, ``Classical simulation of quantum computation, the
  {G}ottesman-{K}nill theorem, and slightly beyond,'' {\em Quantum Information
  \& Computation}, vol.~10, no.~3-4, pp.~0258--0271, 2010.

\bibitem{jozsa2014classical}
R.~Jozsa and M.~{Van den Nest}, ``{Classical simulation complexity of extended
  Clifford circuits},'' {\em Quantum Information {\&} Computation}, vol.~14,
  no.~7{\&}8, pp.~633--648, 2014.

\bibitem{koh2017further}
D.~E. Koh, ``{Further extensions of Clifford circuits and their classical
  simulation complexities},'' {\em Quantum Information {\&} Computation},
  vol.~17, no.~3{\&}4, pp.~262--282, 2017.

\bibitem{bouland2018complexity}
A.~Bouland, J.~F. Fitzsimons, and D.~E. Koh, ``{Complexity Classification of
  Conjugated Clifford Circuits},'' in {\em 33rd Computational Complexity
  Conference (CCC 2018)} (R.~A. Servedio, ed.), vol.~102 of {\em Leibniz
  International Proceedings in Informatics (LIPIcs)}, (Dagstuhl, Germany),
  pp.~21:1--21:25, Schloss Dagstuhl--Leibniz-Zentrum fuer Informatik, 2018.

\bibitem{yoganathan2019quantum}
M.~Yoganathan, R.~Jozsa, and S.~Strelchuk, ``Quantum advantage of unitary
  {C}lifford circuits with magic state inputs,'' {\em Proceedings of the Royal
  Society A}, vol.~475, no.~2225, p.~20180427, 2019.

\bibitem{bravyi2016trading}
S.~Bravyi, G.~Smith, and J.~A. Smolin, ``Trading classical and quantum
  computational resources,'' {\em Phys. Rev. X}, vol.~6, p.~021043, Jun 2016.

\bibitem{bravyi2019simulation}
S.~Bravyi, D.~Browne, P.~Calpin, E.~Campbell, D.~Gosset, and M.~Howard,
  ``Simulation of quantum circuits by low-rank stabilizer decompositions,''
  {\em {Quantum}}, vol.~3, p.~181, Sept. 2019.

\bibitem{howard2017application}
M.~Howard and E.~Campbell, ``Application of a resource theory for magic states
  to fault-tolerant quantum computing,'' {\em Phys. Rev. Lett.}, vol.~118,
  p.~090501, Mar 2017.

\bibitem{seddon2021quantifying}
J.~R. Seddon, B.~Regula, H.~Pashayan, Y.~Ouyang, and E.~T. Campbell,
  ``Quantifying quantum speedups: Improved classical simulation from tighter
  magic monotones,'' {\em PRX Quantum}, vol.~2, p.~010345, Mar 2021.

\bibitem{seddon2019quantifying}
J.~R. Seddon and E.~T. Campbell, ``Quantifying magic for multi-qubit
  operations,'' {\em Proc. R. Soc. A.}, vol.~475, 2019.

\bibitem{wang2019quantifying}
X.~Wang, M.~M. Wilde, and Y.~Su, ``Quantifying the magic of quantum channels,''
  {\em New Journal of Physics}, vol.~21, p.~103002, Oct 2019.

\bibitem{bu2019efficient}
K.~Bu and D.~E. Koh, ``Efficient classical simulation of {C}lifford circuits
  with nonstabilizer input states,'' {\em Phys. Rev. Lett.}, vol.~123,
  p.~170502, Oct 2019.

\bibitem{bu2022classical}
K.~Bu and D.~E. Koh, ``Classical simulation of quantum circuits by half {G}auss
  sums,'' {\em Commun. Math. Phys.}, vol.~390, pp.~471--500, Mar 2022.

\bibitem{liu2020many}
Z.-W. Liu and A.~Winter, ``Many-body quantum magic,'' {\em arXiv preprint
  arXiv:2010.13817}, 2020.

\bibitem{friedgut1996every}
E.~Friedgut and G.~Kalai., ``Every monotone graph property has a sharp
  threshold,'' {\em Proceedings of the American mathematical Society},
  vol.~124, pp.~2993--3002, Aug 1996.

\bibitem{mansour1994learning}
Y.~Mansour, ``Learning {B}oolean functions via the {F}ourier transform,'' in
  {\em Theoretical advances in neural computation and learning}, pp.~391--424,
  Springer, 1994.

\bibitem{lostaglio2015quantum}
M.~Lostaglio, K.~Korzekwa, D.~Jennings, and T.~Rudolph, ``Quantum coherence,
  time-translation symmetry, and thermodynamics,'' {\em Phys. Rev. X}, vol.~5,
  p.~021001, Apr 2015.

\bibitem{lostaglio2015description}
M.~Lostaglio, D.~Jennings, and T.~Rudolph, ``Description of quantum coherence
  in thermodynamic processes requires constraints beyond free energy,'' {\em
  Nature communications}, vol.~6, no.~1, pp.~1--9, 2015.

\bibitem{plenio2008dephasing}
M.~B. Plenio and S.~F. Huelga, ``Dephasing-assisted transport: quantum networks
  and biomolecules,'' {\em New Journal of Physics}, vol.~10, p.~113019, Nov
  2008.

\bibitem{lloyd2011quantum}
S.~Lloyd, ``Quantum coherence in biological systems,'' {\em Journal of Physics:
  Conference Series}, vol.~302, p.~012037, Jul 2011.

\bibitem{levi2014quantitative}
F.~Levi and F.~Mintert, ``A quantitative theory of coherent delocalization,''
  {\em New Journal of Physics}, vol.~16, p.~033007, Mar 2014.

\bibitem{aberg2006quantifying}
J.~Aberg, ``Quantifying superposition,'' {\em arXiv preprint quant-ph/0612146},
  2006.

\bibitem{baumgratz2014quantifying}
T.~Baumgratz, M.~Cramer, and M.~B. Plenio, ``Quantifying coherence,'' {\em
  Phys. Rev. Lett.}, vol.~113, p.~140401, Sep 2014.

\bibitem{winter2016operational}
A.~Winter and D.~Yang, ``Operational resource theory of coherence,'' {\em Phys.
  Rev. Lett.}, vol.~116, p.~120404, Mar 2016.

\bibitem{bu2017maximum}
K.~Bu, U.~Singh, S.-M. Fei, A.~K. Pati, and J.~Wu, ``Maximum relative entropy
  of coherence: An operational coherence measure,'' {\em Phys. Rev. Lett.},
  vol.~119, p.~150405, Oct 2017.

\bibitem{Streltsov2017colloquium}
A.~Streltsov, G.~Adesso, and M.~B. Plenio, ``Colloquium: Quantum coherence as a
  resource,'' {\em Rev. Mod. Phys.}, vol.~89, p.~041003, Oct 2017.

\bibitem{bischof2019resource}
F.~Bischof, H.~Kampermann, and D.~Bru\ss{}, ``Resource theory of coherence
  based on positive-operator-valued measures,'' {\em Phys. Rev. Lett.},
  vol.~123, p.~110402, Sep 2019.

\bibitem{marien2016entanglement}
M.~Mari{\"e}n, K.~M. Audenaert, K.~Van~Acoleyen, and F.~Verstraete,
  ``Entanglement rates and the stability of the area law for the entanglement
  entropy,'' {\em Communications in Mathematical Physics}, vol.~346, no.~1,
  pp.~35--73, 2016.

\bibitem{dwork2014algorithmic}
C.~Dwork and A.~Roth, ``{The Algorithmic Foundations of Differential
  Privacy},'' {\em Found. Trends. Theor. Comput. Sci.}, vol.~9, pp.~211--407,
  Aug. 2014.

\bibitem{dwork2016calibrating}
C.~Dwork, F.~McSherry, K.~Nissim, and A.~Smith, ``Calibrating noise to
  sensitivity in private data analysis,'' {\em Journal of Privacy and
  Confidentiality}, vol.~7, no.~3, pp.~17--51, 2016.

\bibitem{bousquet2002stability}
O.~Bousquet and A.~Elisseeff, ``Stability and generalization,'' {\em The
  Journal of Machine Learning Research}, vol.~2, pp.~499--526, 2002.

\bibitem{bousquet2020sharper}
O.~Bousquet, Y.~Klochkov, and N.~Zhivotovskiy, ``Sharper bounds for uniformly
  stable algorithms,'' in {\em Proceedings of Thirty Third Conference on
  Learning Theory} (J.~Abernethy and S.~Agarwal, eds.), vol.~125 of {\em
  Proceedings of Machine Learning Research}, pp.~610--626, PMLR, 09--12 Jul
  2020.

\bibitem{zhou2017differential}
L.~Zhou and M.~Ying, ``Differential privacy in quantum computation,'' in {\em
  2017 IEEE 30th Computer Security Foundations Symposium (CSF)}, pp.~249--262,
  2017.

\bibitem{aaronson2019gentle}
S.~Aaronson and G.~N. Rothblum, ``Gentle measurement of quantum states and
  differential privacy,'' in {\em Proceedings of the 51st Annual ACM SIGACT
  Symposium on Theory of Computing}, STOC 2019, (New York, NY, USA),
  pp.~322--333, Association for Computing Machinery, 2019.

\bibitem{banchi2021generalization}
L.~Banchi, J.~Pereira, and S.~Pirandola, ``Generalization in quantum machine
  learning: A quantum information standpoint,'' {\em PRX Quantum}, vol.~2,
  p.~040321, Nov 2021.

\bibitem{caro2021generalization}
M.~C. Caro, H.-Y. Huang, M.~Cerezo, K.~Sharma, A.~Sornborger, L.~Cincio, and
  P.~J. Coles, ``Generalization in quantum machine learning from few training
  data,'' {\em arXiv preprint arXiv:2111.05292}, 2021.

\bibitem{bu2021statistical}
K.~Bu, D.~E. Koh, L.~Li, Q.~Luo, and Y.~Zhang, ``On the statistical complexity
  of quantum circuits,'' {\em arXiv preprint arXiv:2101.06154}, 2021.

\bibitem{bu2021effects}
K.~Bu, D.~E. Koh, L.~Li, Q.~Luo, and Y.~Zhang, ``Effects of quantum resources
  on the statistical complexity of quantum circuits,'' {\em arXiv preprint
  arXiv:2102.03282}, 2021.

\bibitem{bu2021rademacher}
K.~Bu, D.~E. Koh, L.~Li, Q.~Luo, and Y.~Zhang, ``Rademacher complexity of noisy
  quantum circuits,'' {\em arXiv preprint arXiv:2103.03139}, 2021.

\bibitem{caro2022out}
M.~C. Caro, H.-Y. Huang, N.~Ezzell, J.~Gibbs, A.~T. Sornborger, L.~Cincio,
  P.~J. Coles, and Z.~Holmes, ``Out-of-distribution generalization for learning
  quantum dynamics,'' {\em arXiv preprint arXiv:2204.10268}, 2022.

\bibitem{gibbs2022dynamical}
J.~Gibbs, Z.~Holmes, M.~C. Caro, N.~Ezzell, H.-Y. Huang, L.~Cincio, A.~T.
  Sornborger, and P.~J. Coles, ``Dynamical simulation via quantum machine
  learning with provable generalization,'' {\em arXiv preprint
  arXiv:2204.10269}, 2022.

\bibitem{aaronson2004improved}
S.~Aaronson and D.~Gottesman, ``Improved simulation of stabilizer circuits,''
  {\em Phys. Rev. A}, vol.~70, p.~052328, Nov 2004.

\bibitem{audenaert2014quantum}
K.~M. Audenaert, ``Quantum skew divergence,'' {\em Journal of Mathematical
  Physics}, vol.~55, no.~11, p.~112202, 2014.

\bibitem{kelman2020towards}
E.~Kelman, G.~Kindler, N.~Lifshitz, D.~Minzer, and M.~Safra, ``Towards a proof
  of the {F}ourier-entropy conjecture?,'' {\em Geometric and Functional
  Analysis}, vol.~30, pp.~1097--1138, Aug 2020.

\bibitem{chakraborty2016upper}
S.~Chakraborty, R.~Kulkarni, S.~V. Lokam, and N.~Saurabh, ``Upper bounds on
  {F}ourier entropy,'' {\em Theoretical Computer Science}, vol.~654,
  pp.~92--112, 2016.
\newblock Computing and Combinatorics.

\bibitem{klivans2010mansour}
A.~R. Klivans, H.~K. Lee, and A.~Wan, ``Mansour's conjecture is true for random
  {DNF} formulas,'' in {\em COLT}, pp.~368--380, Citeseer, 2010.

\bibitem{donnell2013composition}
R.~O'Donnell and L.-Y. Tan, ``A composition theorem for the {F}ourier
  entropy-influence conjecture,'' in {\em Automata, Languages, and Programming}
  (F.~V. Fomin, R.~Freivalds, M.~Kwiatkowska, and D.~Peleg, eds.), (Berlin,
  Heidelberg), pp.~780--791, Springer Berlin Heidelberg, 2013.

\bibitem{donnell2011fourier}
R.~O'Donnell, J.~Wright, and Y.~Zhou, ``The {F}ourier entropy--influence
  conjecture for certain classes of {B}oolean functions,'' in {\em Automata,
  Languages and Programming} (L.~Aceto, M.~Henzinger, and J.~Sgall, eds.),
  (Berlin, Heidelberg), pp.~330--341, Springer Berlin Heidelberg, 2011.

\bibitem{shalev2018fourier}
G.~Shalev, ``On the {F}ourier entropy influence conjecture for extremal
  classes,'' {\em arXiv preprint arXiv:1806.03646}, 2018.

\bibitem{wan2014decision}
A.~Wan, J.~Wright, and C.~Wu, ``Decision trees, protocols and the
  entropy-influence conjecture,'' in {\em Proceedings of the 5th Conference on
  Innovations in Theoretical Computer Science}, ITCS '14, (New York, NY, USA),
  p.~67–80, Association for Computing Machinery, 2014.

\bibitem{gopalan2016degree}
P.~Gopalan, R.~A. Servedio, and A.~Wigderson, ``{Degree and Sensitivity: Tails
  of Two Distributions},'' in {\em 31st Conference on Computational Complexity
  (CCC 2016)} (R.~Raz, ed.), vol.~50 of {\em Leibniz International Proceedings
  in Informatics (LIPIcs)}, (Dagstuhl, Germany), pp.~13:1--13:23, Schloss
  Dagstuhl--Leibniz-Zentrum fuer Informatik, 2016.

\bibitem{gross2006hudson}
D.~Gross, ``Hudson's theorem for finite-dimensional quantum systems,'' {\em
  Journal of Mathematical Physics}, vol.~47, no.~12, p.~122107, 2006.

\end{thebibliography}

\end{document}